\documentclass[twoside,12pt,reqno]{amsart}
\usepackage{amsmath,amsthm, amsfonts,amssymb}
\usepackage{color,soul,mathtools}
\usepackage[dvipsnames]{xcolor}
\usepackage[all]{xy}
\usepackage{graphicx}
\usepackage{indentfirst}
\usepackage{bm}
\usepackage{mathrsfs}
\usepackage{latexsym}
\usepackage{hyperref}
\hypersetup{
    colorlinks = true,
    linkcolor = blue,
    urlcolor = blue,
    citecolor = blue,
    anchorcolor = black,
}
\usepackage{microtype}	
\usepackage{float}
\usepackage{bookmark}
\usepackage{tikz}[2010/10/13]
\input{Qcircuit}
\usepackage{mathtools}
\usepackage{ifpdf}
\usepackage{doi}

\newtheorem{Thm}{Theorem}[section]

\newtheorem{proposition}[Thm]{\bf Proposition}
\newtheorem{corollary}[Thm]{\bf Corollary}
\newtheorem{lemma}[Thm]{\bf Lemma}

\newtheorem{remark}[Thm]{Remark}
\newtheorem{definition}[Thm]{\bf Definition}

\newtheorem{theorem}[Thm]{\bf Theorem}

\newcommand{\abs}[1]{\left\vert #1 \right\vert}
\DeclareMathOperator{\Tr}{tr}
\newcommand{\be}{\begin{equation}}
\newcommand{\ee}{\end{equation}}

\newcommand{\beq}{\begin{eqnarray}}
\newcommand{\eeq}{\end{eqnarray}}
\newcommand{\beqs}{\begin{eqnarray*}}
\newcommand{\eeqs}{\end{eqnarray*}}
\newcommand{\bmm}[4]{\begin{pmatrix} #1&#2\\#3&#4 \end{pmatrix}}
\newcommand{\E}{\mathcal{E}}
\newcommand{\FS}{\mathfrak{F}_{\text{{s}}}}
\newcommand{\Max}{\ket{\rm{Max}}}
\newcommand{\GHZ}{\ket{\rm{GHZ}}}
\newcommand{\bR}{{\textbf{R}}}
\newcommand{\bL}{{\textbf{L}}}
\newcommand{\bS}{{\textbf{S}}}
\newcommand{\hrefdoi}[2]{\href{https://doi.org/#1}{#2}}

\renewcommand{\thesection}{\arabic{section}}

\renewcommand{\theequation}{\arabic{section}.\arabic{equation}}

\tikzstyle WL=[line width=5pt,opacity=1.0]

\newcommand{\drawWL}[3]
{
    \draw[white,WL]  (#2) -- (#3);
    \draw[#1] (#2) -- (#3);
}

\newcommand{\mn}{/-3}
\newcommand{\nn}{/3}
\newcommand{\size}[1]{\fontsize{10pt}{\baselineskip}\selectfont{#1}}
\newcommand{\Z}{\mathbb{Z}}

\newcommand{\fdoublemeasure}[6]
{
\fmeasure{#1}{#2}{3*#3}{#4}{}
\fmeasure{#1-2*#3}{#2}{#3}{#4}{#5}
\node at (#1-3*#3, #2-#4+2*#3) {\size{$#6$}};
}

\newcommand{\fmeasure}[5]{
\node at (#1-#3,#2-#4) {\size{$#5$}};
\draw (#1,#2) --(#1,#2-#4)  arc (0:180:#3) -- (#1-#3-#3,#2);
}

\newcommand{\fqudit}[5]{
\node at (#1-#3,#2+#4) {\size{$#5$}};
\draw (#1,#2) --(#1,#2+#4)  arc (0:-180:#3) -- (#1-#3-#3,#2);
}

\newcommand{\fsqudit}[5]{
\node at (#1-#3,#2+#4) {\size{$#5$}};
\draw (#1,#2) --(#1,#2+#4)  arc (0:180:#3) -- (#1-#3-#3,#2);
}

\newcommand{\fdoublequdit}[6]
{
\fqudit{#1}{#2}{3*#3}{#4}{}
\fqudit{#1-2*#3}{#2}{#3}{#4}{#5}
\node at (#1-3*#3, #2+#4+2*#3) {\size{$#6$}};
}

\newcommand{\fbraid}[4]{
\draw (#1,#2)--(#3,#4);
\draw (#1,#4)--(2/3*#1+1/3*#3,2/3*#4+1/3*#2);
\draw (#3,#2)--(2/3*#3+1/3*#1,2/3*#2+1/3*#4);
}

\newcommand{\braida}[4]
{
\draw (#1+2*#3,0+#2)--(#1+2*#3, 8*#4+#2);
\draw (#1+8*#3,0+#2)--(#1+8*#3, 8*#4+#2);
\draw (#1+0*#3,8*#4+#2)--(#1+1.5*#3,6.5*#4+#2);
\draw (#1+2.5*#3,5.5*#4+#2)-- (#1+3.5*#3,4.5*#4+#2);
\draw (#1+4.5*#3,3.5*#4+#2)-- (#1+6*#3,2*#4+#2) --(#1+6*#3,0*#4+#2);
\draw (#1+0,0*#4+#2)--(#1+1.5*#3,1.5*#4+#2);
\draw (#1+2.5*#3,2.5*#4+#2)--(#1+6*#3,6*#4+#2) --(#1+6*#3,8*#4+#2);
}

\newcommand{\braidb}[4]
{
\draw (#1+2*#3,-0+#2)--(#1+2*#3,- 8*#4+#2);
\draw (#1+8*#3,-0+#2)--(#1+8*#3,- 8*#4+#2);
\draw (#1+0*#3,-8*#4+#2)--(#1+1.5*#3,-6.5*#4+#2);
\draw (#1+2.5*#3,-5.5*#4+#2)-- (#1+3.5*#3,-4.5*#4+#2);
\draw (#1+4.5*#3,-3.5*#4+#2)-- (#1+#1,-2*#4+#2) --(#1+#1,-0*#4+#2);
\draw (#1+0,-0*#4+#2)--(#1+1.5*#3,-1.5*#4+#2);
\draw (#1+2.5*#3,-2.5*#4+#2)--(#1+#1,-6*#4+#2) --(#1+#1,-8*#4+#2);
}

\newcommand{\fgraphd}[2]
{\tikz{
\braida{0}{0}{1\mn}{1\nn}
\braidb{6\mn}{16\nn}{1\mn}{1\nn}
\fqudit{12\mn}{16\nn}{-1\mn}{0\nn}{}
\fqudit{6\mn}{16\nn}{-1\mn}{1\nn}{}
\fqudit{0\mn}{8\nn}{-1\mn}{10\nn}{i_A}
\fmeasure{6\mn}{0\nn}{-1\mn}{1\nn}{-m_1}
\node at (13\mn,-4\nn) {\size{$m_1$}};
\draw (12\mn,8\nn) -- (12\mn,-6\nn);
\draw (14\mn,8\nn) -- (14\mn,-8\nn);
\draw (13\mn,-6\nn) --(11\mn,-6\nn) -- (11\mn,-8\nn) --(13\mn,-8\nn);
\draw (15\mn,-6\nn) --(21\mn,-6\nn) -- (21\mn,-8\nn) --(15\mn,-8\nn);
\node at (12\mn, -7\nn) {$T$};
\draw (16\mn,-6\nn) -- (16\mn,-4\nn);
\node at (18\mn,-5\nn) {$\cdots$};
\draw (20\mn,-6\nn) -- (20\mn,-4\nn);
\draw (16\mn,-8\nn) -- (16\mn,-10\nn);
\node at (18\mn,-9\nn) {$\cdots$};
\draw (20\mn,-8\nn) -- (20\mn,-10\nn);
\fmeasure{12\mn}{-8\nn}{-1\mn}{2\nn}{-m_2}
\node at (1\mn,-13\nn) {\size{$m_2$}};
\draw (0,0) -- (0,-14\nn);
\draw (2\mn,0) -- (2\mn,-14\nn);
}}

\title[Holographic Software]{Holographic Software\\ for Quantum Networks}
\address{17 Oxford Street, Harvard University, Cambridge, MA 02138, USA}
\email{jaffe@g.harvard.edu}
\email{zhengweiliu@fas.harvard.edu}
\email{airwozz@gmail.com}
\author{Arthur Jaffe}
\author{Zhengwei Liu}
\author{Alex Wozniakowski}

\begin{document}
\begin{abstract}
 We introduce a  pictorial approach to quantum information, called {\it holographic software}. Our software captures both algebraic and topological aspects of quantum networks.  It yields a bi-directional dictionary to translate between a topological approach and an algebraic approach.  Using our software, we give a topological simulation for quantum networks.
The string Fourier transform (SFT) is our basic tool to transform product states into states with maximal entanglement entropy.  We obtain a pictorial interpretation of  Fourier transformation,  of measurements,  and of local transformations, including the $n$-qudit Pauli matrices and their representation by Jordan-Wigner transformations.
 We use our software to discover interesting new protocols for multipartite communication. In summary, we build a bridge linking the theory of planar para algebras with quantum information.
\end{abstract}

\maketitle

\tableofcontents

\setcounter{section}{0}
\section{Introduction}
In classical information science and computer science the elementary unit of information takes two  possible values; it is called a bit. 
In quantum information one studies vectors in a Hilbert space $\mathcal{H}$, and these vectors encode information. The complex Hilbert space corresponding to a bit is $2$-dimensional, and a vector in this space is called  a qubit.  
Here we study the analog for a $d$-valued classical system, in which the $d$-dimensional vector in quantum information is called a qudit.  

We are interested in the Hilbert space $\mathcal{H}$ for $n$ qudits, defined as a twisted,  $n$-fold tensor product of the $1$-qudit space, and a vector in this space is called an $n$-qudit. Unitary transformations on $\mathcal{H}$ play an important role, as do measurements (dual vectors in one of the $1$-qudit spaces in the tensor product).  

We study quantum information using the pictorial framework PAPPA (or parafermion, planar, para algebra)  introduced in \S3 of~\cite{JL}, and the related string Fourier transform (SFT) analyzed in \S6 of that paper.  
We propose a new charged-string language, in which we represent multipartite entanglement, Pauli matrices, measurements,  and quantum communication protocols by elementary pictures.
We label our strings by an integer charge $k\in\mathbb{Z}_{d}$.  
 Using insights from this language, we have designed a new protocol~\cite{ConstructiveSimulation}.  

In this paper we give the mathematical basis for our language. 
We represent an elementary transformation by a box with inputs and outputs that we denote by strings. One forms an algebra of transformations by connecting the output strings to input strings, as in category theory.
However, our approach differs in important ways from previous pictorial  studies of quantum information \cite{KL-1}, \cite{AC,Coeckebook}.  Three such major differences are:
\begin{itemize}
\item{} We resolve  a single string into two strings, giving more structure to transformations.
\item{}  We assign  a $\mathbb{Z}_{d}$-valued charge to each string.
\item{} We use the string Fourier transform as the central tool to produce entanglement.
\end{itemize}

The resolution into two strings arises naturally in mathematics in the context of subfactor theory, as compared with representation theory. In physics it arises in the context of string theory diagrams compared with Feynman diagrams. 

We follow the usual conventions of planar algebras and subfactor theory \cite{Jon98}.  Multiplication, described algebraically by ordering factors from left to right,  is represented by composing pictures vertically from bottom to top.  Convolution is pictured by composing pictures horizontally from left to right.   The adjoint of a picture is  defined as an anti-linear, anti-isomorphism given by vertical reflection.  

The assignment of charge to a string arises from the interpretation of a string as carrying $0,1,\ldots,d-1$  ``parafermions,'' each with charge $\pm1$ in the additive group $\Z_{d}$.  The charges that we label by $k$ appear in the pictures as labels on the left side of a string. Changes in their placement lead to the notion of  topological para isotopy of pictures, generalizing isotopy. For details, see~\cite{JL}.

We find that the SFT is the natural gate to entangle $n$-qudit product states.   See \cite{JL} for details of how the SFT generalizes the Fourier transform. The SFT  produces entanglement through a single ``one-click pictorial rotation'' of the picture \eqref{Max-kn} that represents the product state. In \S\ref{Sect:Entropy} we show that these states have maximal entanglement entropy, so they provide good generalizations of Bell states of 2-qudits.  

The language we study here has certain advantages. In \S \ref{HoloSoft} and \S \ref{Sect:SFT} we provide  a dictionary to translate bijectively between pictures and quantum information protocols. Here we comment on a few aspects of our pictures that highlight their usefulness:  
\begin{itemize}
\item{} Our space of fundamental states, or  $1$-qudits, is spanned by caps pictured below. These have two string outputs, rather than one. The resolution of one string into two illuminates hidden structures: 
\[
\raisebox{-.25cm}{
\tikz{
   \fqudit {0}{0}{1\mn}{1\nn}{\phantom{-}k}
}}\;.
\]
\item{}  Our $1$-qudit transformations have two strings as inputs and two strings as outputs, rather than  a single string.  
\[
\raisebox{-.45cm}{
\tikz{
\draw (-1/6,1/6) rectangle (1/6+2/3,1-1/6);
\draw (0,0)--(0,1/6);
\draw (2/3,0)--(2/3,1/6);
\draw (0,1)--(0,1-1/6);
\draw (2/3,1)--(2/3,1-1/6);
\node at (1/3,1/2) {\size{$T$}};
}}\;.
\]

\item{} We obtain the braid
\[
\raisebox{-.3cm}{\scalebox{1}{
\tikz{
\fbraid{0}{2\nn}{2\mn}{0}
}}}\ ,
\]
as a linear combination of $1$-qudit transformations, see \S\ref{Sec:braided relations}.

\item{} A charge $k$ can pass freely under our braids, but not over them,
\[
\raisebox{-.4cm}{\scalebox{.8}{
\tikz{
\fbraid{0}{0}{3\mn}{3\nn}
\node at (3\mn,.8\nn) {\size{$k$}};
}}} =
\raisebox{-.4cm}{\scalebox{.8}{
\tikz{
\fbraid{0}{0}{3\mn}{3\nn}
\node at (.9\mn,2.9\nn) {\size{$k$}};
}}}\;.
\]

\item{}  Our pictures satisfy a para isotopy relation, rather than isotopy invariance, under vertical motion of charge. With $q^{d}=1$ is a $d^{\rm th}$ root of unity, 
\[
\raisebox{-.4cm}{\scalebox{.9}{
\begin{tikzpicture}
\draw (0,0) --(0,01);
\draw (0.4,0) --(0.4,01);
\node (0,0) at (-0.15,0.2) {\scriptsize $k$};
\node (0.2,0) at (0.25,.8) {\scriptsize $\ell$};
\end{tikzpicture}
}}
\scriptstyle =\   q^{k\ell}
\raisebox{-.4cm}{\scalebox{.9}{
\begin{tikzpicture}
\draw (0,0) --(0,01);
\draw (0.4,0) --(0.4,01);
\node (0,0) at (-0.15,0.8) {\scriptsize $k$};
\node (0.2,0) at (0.25,.2) {\scriptsize $\ell$};
\end{tikzpicture}
}}
\;.
\]
The interpolation of this relation, with $k$ and $\ell$ at the same height, plays a central role, see \S\ref{Sect:ParaIsotopy}.

\item{}  This leads us to an elementary pictures for $1$-qudit Pauli matrices $X,Y,Z$, 
\[
X=
\raisebox{-.4cm}{
\tikz{
\draw (0,0)--(0,1);
\draw (2/3,0)--(2/3,1);
\node at (1.5/3,1/2) {\scriptsize 1};
}}\;,\quad
Y=
\raisebox{-.4cm}{
\tikz{
\draw (0,0)--(0,1);
\draw (2/3,0)--(2/3,1);
\node at (-.5/3,1/2) {\scriptsize -1};
}}\;,\quad
Z=
\raisebox{-.4cm}{
\tikz{
\draw (0,0)--(0,1);
\draw (2/3,0)--(2/3,1);
\node at (-.5/3,1/2) {\scriptsize 1};
\node at (1.5/3,1/2) {\scriptsize -1};
}}\;.
\]

\item{} The Hilbert space of multiple qudits is spanned by a tensor product of $1$-qudits,
\[
\raisebox{-.4cm}{\scalebox{.9}{
 \raisebox{-.5cm}{\tikz{
\fqudit{0\mn}{0\nn}{1\mn}{2\nn}{\phantom{k} k_2}
\fqudit{4\mn}{0\nn}{1\mn}{3\nn}{\phantom{k} k_1}
\fqudit{-6\mn}{0\nn}{1\mn}{1\nn}{\phantom{k} k_n}
\node at (-4\mn,1\nn) {$\cdots$};
}}}\;.
}\] 

\item{}
We give the corresponding pictures for the Jordan-Wigner representation of  $n$-qudit Pauli matrices in Proposition \ref{Prop:JWasLocal}. 

\item{} The SFT acts on pictures by cyclically permuting the strings exiting a picture. On vectors the  SFT is a unitary transformation that we denote as $\mathfrak{F}_{s}$.   Acting on an $n$-qudit product state,  
\be\label{Max-kn}
\FS \,
\scalebox{.5}{
 \raisebox{-.5cm}{\tikz{
\fqudit{0\mn}{0\nn}{1\mn}{2\nn}{\phantom{k} k_3}
\fqudit{3\mn}{0\nn}{1\mn}{3\nn}{\phantom{k} k_2}
\fqudit{6\mn}{0\nn}{1\mn}{4\nn}{\phantom{k} k_1}
\fqudit{-6\mn}{0\nn}{1\mn}{1\nn}{\phantom{k} k_n}
\node at (-4\mn,1\nn) {$\cdots$};
}} \; \;
}
\ \ = \ \ 
\scalebox{.4}{\raisebox{-1.1cm}{
\tikz{
\fqudit{-2\nn}{0}{9\mn}{0.01\nn}{}
\fqudit{0\nn}{0}{1\mn}{4\nn}{\phantom{ll}k_{1}}
\fqudit{3\nn}{0}{1\mn}{3\nn}{\phantom{ll}k_{2}}
\fqudit{6\nn}{0}{1\mn}{2\nn}{\phantom{ll}k_{3}}
\fqudit{12\nn}{0}{1\mn}{1\nn}{\phantom{ll}k_{n}}
\fmeasure{-2\nn}{0}{-1\nn}{0.05\nn}{}
\node at (10\nn,.5) {\dots};
}}} \; \; \;
= \; \; \;
\scalebox{.4}{\raisebox{-.8cm}{
\tikz{
\fqudit{1.5\nn}{0}{7\mn}{0.01\nn}{}
\node at (1.75\nn,4\nn) {$k_{1}$};
\fqudit{3\nn}{0}{1\mn}{3\nn}{\phantom{ll}k_{2}}
\fqudit{6\nn}{0}{1\mn}{2\nn}{\phantom{ll}k_{3}}
\fqudit{12\nn}{0}{1\mn}{1\nn}{\phantom{ll}k_{n}}
\node at (10\nn,.5) {\dots};
}}} \;.
\ee
One single SFT entangles a product state with an  arbitrary number of qudits. 
In Theorem \ref{Thm:MaxEntangle} we show that any state \eqref{Max-kn} has  maximal entanglement entropy. 

\item{}
Acting on an  $n$-qudit vector, $\FS$ involves $2n$ braids. It has the pictorial representation   
\[
\FS=\overbrace{\underbrace{\scalebox{.7}{
\raisebox{-1.2cm}{\tikz{
\fbraid{4/3}{-2\nn}{2/3}{-1\nn}
\fbraid{2/3}{-1\nn}{0/3}{0\nn}
\node at (1/3,1\nn) {\tiny{$\cdots$}};
\node at (1/3,-2\nn) {\tiny{$\cdots$}};
\fbraid{-0/3}{0\nn}{-2/3}{1\nn}
\fbraid{-2/3}{1\nn}{-4/3}{2\nn}
\fbraid{-4/3}{2\nn}{-6/3}{3\nn}
\fbraid{-6/3}{3\nn}{-8/3}{4\nn}
\draw (4/3,-2\nn) -- (4/3,-3\nn);
\draw (4/3,-1\nn) -- (4/3,5\nn);
\draw (2/3,-2\nn) -- (2/3,-3\nn);
\draw (2/3,0\nn) -- (2/3,5\nn);
\draw (0/3,1\nn) -- (0/3,5\nn);
\draw (0/3,-1\nn) -- (0/3,-3\nn);
\draw (-2/3,2\nn) -- (-2/3,5\nn);
\draw (-2/3,0\nn) -- (-2/3,-3\nn);
\draw (-4/3,3\nn) -- (-4/3,5\nn);
\draw (-4/3,1\nn) -- (-4/3,-3\nn);
\draw (-6/3,4\nn) -- (-6/3,5\nn);
\draw (-6/3,2\nn) -- (-6/3,-3\nn);
\draw (-8/3,4\nn) to [bend right=45] (-8/3,3\nn);
}}}}_{2\mbox{n} \; \mbox{strings}}}^{2\mbox{n} \; \mbox{strings}}
\;.
\]
In Theorem \ref{Thm:SFT} we compute the matrix  elements of $\FS$ acting on $n$-qudit vectors. 
 
 \item{}
The SFT rotates transformations of 1-qudits by $90^{\circ}$; in fact one defines the SFT as a cyclic permutation of the string labels.   
\[
\raisebox{-.5cm}{
\tikz{
\node at (-1/3 + -1/6,1/2) {\size{$\mathfrak{F}_{s}$}};
\draw (-1/6,1/6) rectangle (1/6+2/3,1-1/6);
\draw (0,-1/6)--(0,1/6);
\draw (2/3,-1/6)--(2/3,1/6);
\draw (0,1+1/6)--(0,1-1/6);
\draw (2/3,1+1/6)--(2/3,1-1/6);
\node at (1/3,1/2) {\size{$T$}};
}}
=
\raisebox{-.5cm}{
\tikz{
\draw (-1/6,1/6) rectangle (1/6+2/3,1-1/6);
\draw (0,-1/6)--(0,1/6);
\draw (2/3,-1/6)--(2/3,1/6);
\draw (0,1+1/6)--(0,1-1/6);
\draw (2/3,1+1/6)--(2/3,1-1/6);
\node at (1/3,1/2) {\size{$\mathfrak{F}_{s} T$}};
}}
=
\raisebox{-.6cm}{\scalebox{.7}[.8]{
\tikz{
\draw (-1/6,1/6) rectangle (1/6+2/3,1-1/6);
\draw (0,-1/6) arc (0:-180:.25);
\draw (2/3,1+1/6) arc (180:0:.25);
\draw (-.5,-1/6)--(-.5,1+1/6);
\draw (2/3+.5,-1/6)--(2/3+.5,1+1/6);
\draw (0,-1/6)--(0,1/6);
\draw (2/3,-1/6)--(2/3,1/6);
\draw (0,1+1/6)--(0,1-1/6);
\draw (2/3,1+1/6)--(2/3,1-1/6);
\node at (1/3,1/2) {\size{$T$}};
}}}\;.
\]
The $\FS$ does not act on transformations as a homomorphism.  The SFT sends a product of neutral 1-qudit transformations to the convolution of two SFT's,
\[
\begin{tikzpicture}
\begin{scope}[scale=.5]
\foreach \x in {0,1}{
\foreach \y in {0,1}{
\foreach \u in {0}{
\foreach \v in {1,2,3}{
\coordinate (A\u\v\x\y) at (\x+1.5*\u,\y+3*\v);
}}}}

\foreach \u in {0}{
\foreach \v in {1,2}{
\draw (A\u\v00) rectangle (A\u\v11);
\node at (-1,5) {\size{$\mathfrak{F}_{s}$}};
\node at (0.5,3.5) {\size{$T$}};
\node at (0.5,6.5) {\size{$S$}};
}}

\draw (0,3)--++(-.5,-.5);
\draw (1,3)--++(.5,-.5);
\draw (0,7)--++(-.5,.5);
\draw (1,7)--++(.5,.5);

\draw (A0101) to [bend left=30] (A0200);

\draw (A0111) to [bend left=-30] (A0210);
\node at (3,5) {$=$};
\end{scope}
\end{tikzpicture}
\raisebox{.5cm}{\begin{tikzpicture}
\begin{scope}[shift={(16,4.5)},rotate=90,scale=.7]
\foreach \x in {0,1}{
\foreach \y in {0,1}{
\foreach \u in {0}{
\foreach \v in {1,2,3}{
\coordinate (A\u\v\x\y) at (\x+1.5*\u,\y+3*\v);
}}}}

\foreach \u in {0}{
\foreach \v in {1,2}{
\draw (A\u\v00) rectangle (A\u\v11);
\node at (0.5,6.5) {\size{$\mathfrak{F}_{s} T$}};
\node at (0.5,3.5) {\size{$\mathfrak{F}_{s} S$}};
}}

\draw (0,3)--++(-.5,-.5);
\draw (1,3)--++(.5,-.5);
\draw (0,7)--++(-.5,.5);
\draw (1,7)--++(.5,.5);

\draw (A0101) to [bend left=30] (A0200);

\draw (A0111) to [bend left=-30] (A0210);
\end{scope}
\end{tikzpicture}}
\]
Interesting inequalities related to uncertainty emerge from the last relation, see \cite{Liu-NCUP,Liu-Wang-Wu,Liu-NCUP-2}. 
\end{itemize}

We describe these and other phenomena in quantum information by elementary pictures, which earlier diagrammatic approaches to quantum information did not capture naturally. 
We introduced preliminary versions of this paper and the companion paper \cite{ConstructiveSimulation} as \cite{QuditIsotopy,CompressedTeleportation}, but we do not plan to publish those papers in a journal. 

Later we extended the results of this paper in  \cite{QuonLanguage}, by resolving  two string 1-qudits into four string 1-qudits. We based the four-string picture on viewing the 1-qudit as a neutral pair composed of a charge and an anti-charge.  One advantage of such a picture is that para-isotopy of neutral composites reduces to  isotopy.      
This generalization from two to four strings  led us naturally  from planar pictures, to pictures in three-dimensional space, see~\cite{QuonLanguage}. There we demonstrated, among other things, the topological nature of the quantum controlled-NOT gate (CNOT gate, Feynman gate). This includes a natural 3D interpretation of the SFT as a $90^{\circ}$ rotation in a 2-plane. We picture the CNOT as two linked channels transmitting  information in 3-space, and we present an insightful 3D picture of teleportation.  
\smallskip

\paragraph{\bf Some Context: }
Manin and Feynman introduced the concept of  quantum simulation for quantum systems \cite{Manin-book,Feynman,Manin}. One can use an abstract  language \textbf{L}  with words and grammar, along with a simulation  $\bS:\textbf{L} \mapsto \bR$ to map onto some interesting mathematical area  \textbf{R}, as discussed in \cite{PictureLanguage}. 

By \textit{software} we mean that we can simulate quantum information {\bR} by a picture language  \bL.  Throughout this paper, \textbf{L} is the picture language which we call PAPPA \cite{JL}, while \textbf{R} is quantum information.   We introduce a simulation \textbf{S} so that we can understand computations in quantum information through PAPPA.  So we also call PAPPA ``software for quantum information.''  We call this software \textit{holographic}, since we can implement any  computation in \textbf{L}, without referring to \bR.  All pictures in \textbf{L} map to qudit transformations in \textbf{R} under the simulation \bS.  

The simulation \textbf{S} of elementary pictures in \textbf{L} leads us to concepts  that may have no physical meaning  in \bR. We call them virtual concepts in \bR, as discussed in \cite{PictureLanguage}. In this paper we introduce one additional notion inspired by quantum information, that one does not find in PAPPA: a dotted line dividing pictures that correspond to different physical regions, see~\S\ref{Sect:DottedLine}. This notion and virtual concepts are crucial to understanding entanglement and the design of protocols by topological isotopy in \bL, as we discuss in the section \textit{Topological Simulation} of   \ref{sect:Topology-Algebra} and the companion paper~\cite{ConstructiveSimulation}.  

\setcounter{equation}{0}
\section{Basic Algebraic Notation}
\subsection{Qudits}
A $1$-qudit is a vector  in a $d$-dimensional Hilbert space, where $d$ is the \textit{degree} of the qudit.  (The usual case of qubits corresponds to $d=2$.)  We denote an orthonormal basis using Dirac notation by $\ket{k}$. We call $k$ the charge of the qudit, and generally $k\in\Z_{d}$, the cyclic group of order $d$.

The dual $1$-qudit $\bra{\ell}$ is a vector  in the space dual to the $d$-dimensional Hilbert space. And $\langle \ell | k \rangle=\delta_{\ell,k}$,
where $\delta_{\ell, k}$ is the Kronecker delta.

The $n$-qudit space $\mathcal{H}$ is  the twisted $n$-fold tensor product of the $1$-qudit space.  An orthonormal basis for $n$-qudits is $|\vec k\rangle =\ket{k_{1},k_{2},\ldots, k_{n}}$, where this ket has total charge $\vert\vec k\vert=k_{1}+k_{2}+\cdots + k_{n}$.  The dual basis is $\langle\vec \ell |$.  Every linear transformation on $n$-qudits can be written as a sum of the $d^{2n}$ homogeneous transformations
	\be
		M_{\vec \ell, \vec k}=|\vec\ell \,\rangle\,\langle \vec k \, | \;,
		\quad\text{with charge}\quad
		\vert\vec k\vert-\vert\vec \ell\vert\;.
	\ee
The matrix elements of $T=\sum\limits_{\vec k,\vec \ell}\  t_{\vec \ell,\vec k} \ M_{\vec \ell, \vec k}$ are just $t_{\vec \ell,\vec k} =\langle\vec \ell  |  T |  \vec k \,\rangle$.

\subsection{The parafermion algebra\label{Sect:PFA}}

The parafermion algebra has been widely studied both in mathematics and in physics.
\begin{definition}
The \textit{parafermion algebra} is a $\ast$-algebra with unitary generators $c_{j}$, which satisfy
\begin{equation}\label{ParafermionAlgebra}
c_{j}^{d}=1 \; \; \; \mbox{and} \; \; \; c_{j}c_{k}=q \, c_{k}c_{j} \; \; \; \mbox{for} \; \; 1 \leqslant  j< k \leqslant m.
\end{equation}
\noindent Here $q \equiv e^{\frac{2 \pi i}{d}}$, $i \equiv \sqrt{-1}$, and $d$ is the order of the parafermion.
\end{definition}
Consequently $c^{\ast}_{j}=c^{-1}_{j}=c^{d-1}_{j}$, where $\mbox{*}$ denotes the adjoint. Majorana fermions arise for $d=2$.

\subsection{Transformations of $1$-qudits}
Let $q^{d}=1$ and  $\zeta=q^{1/2}$ be a square root of $q$ with the property $\zeta^{d^{2}}=1$. Matrices $X,Y,Z, F,G$ play an important role in many areas of quantum information: for example in teleportation protocols \cite{Bennett-etal},  in error-correction  \cite{Gottesman-Clifford}, or in fault-tolerance  \cite{Gottesman-fault-tolerant}. Three of these are the qudit Pauli matrices, with the algebraic definitions
	\be\label{XYZ-Defn-1}
	X\ket{k}=\ket{k+1}\;,\
	Y\ket{k}=\zeta^{1-2k}\ket{k-1}\;,\
	Z\ket{k}=q^{k}\ket{k}\;.	
	\ee
The Fourier matrix $F$ and the Gaussian $G$ are
	\be\label{F-and-G}
	F\ket{k}=\frac{1}{\sqrt{d}}\sum_{\ell=0}^{d-1} q^{k\ell} \ket{\ell}\;,
\quad
G\ket{k}=\zeta^{k^2}\ket{k}\;.
	\ee
These matrices satisfy the relations
	\be\label{XY-qYX}
		XY=qYX\;,\quad
		YZ=qZY\;,\quad
		ZX=qXZ\;,\quad
		XYZ=\zeta\;.
	\ee
	\be
	FXF^{-1}=Z\;,\quad
	GXG^{-1}=Y^{-1}\;.
	\ee
	
\goodbreak	
\subsection{Transformations of $2$-qudits}
\subsubsection{The multipartite entangled resource state}
We represent the multipartite entangled resource state for $2$-qudits as
\[
\ket{\text{Max}}=\displaystyle \frac{1}{\sqrt{d}}\sum_{k=0}^{d-1}\ket{k,-k}\;.
\]
We say it costs 1 edit if two persons use this entangled state in a protocol.

\subsubsection{Controlled gates}
We give the protocol for controlled transformations $C_{1,A}$ in Figure~\ref{Fig:Controlled gate 1} and $C_{A,1}$ in Figure~\ref{Fig:Controlled gate 2}, for different control qudits.
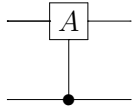
\begin{figure}[H]
\[
\scalebox{0.91}{
\Qcircuit @C=1.5em @R=2em  {
& \gate{A} \qwx[1] \qw & \qw \\
& \control \qw & \qw
}}
\]
\caption{The controlled gate $C_{1,A}$ acts on the $2$-qudit $\ket{k_{1}, k_{2}}$  and gives $C_{1,A}\ket{k_{1},k_{2}}=\ket{k_{1},A^{k_{1}}k_{2}}$. The first qudit $k_{1}$ is the control qudit. \label{Fig:Controlled gate 1}}
\end{figure}

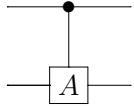
\begin{figure}[H]
\[
\scalebox{0.91}{
\Qcircuit @C=1.5em @R=2em  {
& \control \qw & \qw\\
& \gate{A} \qwx[-1] \qw & \qw
}}
\]
\caption{The controlled gate $C_{A,1}$ acts on the $2$-qudit $\ket{k_{1}, k_{2}}$ and gives $C_{A,1}\ket{k_{1},k_{2}}=\ket{A^{k_{2}}k_{1},k_{2}}$. The second qudit $k_{2}$ is the control qudit. \label{Fig:Controlled gate 2}}
\end{figure}
We sometimes allow more general controlled transformations of the form
\be
T=\sum_{l=0}^{d-1} \ket{\ell}\bra{\ell} \otimes T(\ell),
\ee
where the control is on the first qudit, and $T(\ell)$ can be an arbitrary transformation  on the target qudit. The picture for such a controlled transformation is illustrated in Figure~\ref{Fig:Controlled T}; a corresponding configuration with the second  qudit acting as the control  would also be possible.
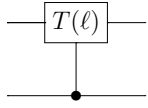
\begin{figure}[h]
\[
\scalebox{0.8}{
\Qcircuit @C=1.5em @R=2em  {
& \gate{T(\ell)} \qwx[1] \qw & \qw\\
& \control \qw & \qw
}}
\]
\caption{Controlled transformations. \label{Fig:Controlled T}}
\end{figure}

The measurement controlled gate is illustrated in  Figure~\ref{Fig:Measurement-controlled-gate}.
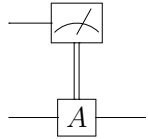
\begin{figure}[h]
\[
\scalebox{0.91}{
\Qcircuit @C=1.5em @R=2em  {
& \meter \cwx[1] \\
& \gate{A} &\qw 
}}
\]
\caption{Measurement controlled gate.}
\label{Fig:Measurement-controlled-gate}
\end{figure}
If the qudit is measured by the meter as $k$, then one applies $A^k$ to the target qudit. It costs 1 cdit to transmit the result, when the two qudits belong to different persons.

\subsection{Qubit case: $d=2$ and $\zeta=+i$}
In the case $d=2$ with $\zeta=+i=\sqrt{-1}$ the $1$-qubit matrices $X,Y,Z$ are the Pauli matrices $\sigma_{x},\sigma_{y},\sigma_{z}$, while $F=H=\frac{1}{\sqrt2}\bmm{1}{1}{1}{-1}$ is the Hadamard matrix, and $G=S=\bmm{1}{0}{0}{i}$ is the phase matrix.  For $2$-qubits, the matrix $C_{1,X}$ is CNOT. The matrices generate the Clifford group.
The elements of the Clifford group can be simulated in polynomial time on a classical computer via the Gottesman-Knill theorem~\cite{GottesmanKnill, Vidal, AaronsonGottesman}. 

\subsection{Simplifying tricks}
We give four elementary algebraic tricks to simplify the algebraic protocols; we illustrate them in Figures. \ref{Fig:Trick1}--\ref{Fig:Trick4}.  Here we consider the Gaussian $G$ defined in \eqref{F-and-G}  for general $d$. 
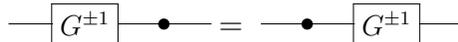
\begin{figure}[H]
\[
\scalebox{0.91}{\raisebox{.15cm}{
\Qcircuit @C=1.5em @R=2em  {
& \gate{G^{\pm 1}} &  \control \qw & \qw
}}}
=
\scalebox{0.91}{\raisebox{.15cm}{
\Qcircuit @C=1.5em @R=2em  {
& \control \qw &  \gate{G^{\pm 1}} & \qw
}}}
\]
\caption{Trick 1:
The control gate commutes with the Gaussian transformation on the control qudit.\label{Fig:Trick1}}
\end{figure}
\begin{figure}[H]
\[
\scalebox{0.91}{\raisebox{.15cm}{
\Qcircuit @C=1.5em @R=2em  {
& \gate{G^{\pm1}} & \meter
}}}
=
\scalebox{0.91}{\raisebox{.15cm}{
\Qcircuit @C=1.5em @R=2em  {
& \meter
}}}
\]
\caption{Trick 2:
The Gaussian transformation does not affect measurement of the meter, so we can remove it.\label{Fig:Trick2}}
\end{figure}
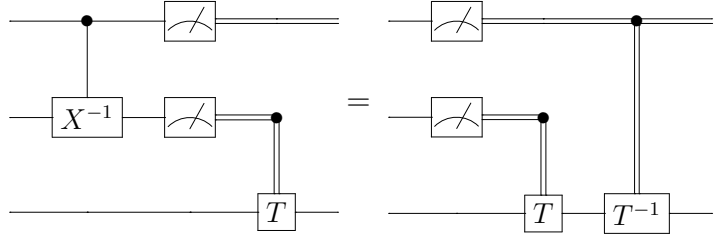
\begin{figure}[H]
\[
\scalebox{0.91}{\raisebox{1.3cm}{
\Qcircuit @C=1.5em @R=2em  {
&  \control \qw        & \meter & \cw &  \cw \\
&   \gate{X^{-1}} \qwx[-1]   & \meter & \control \cw \cwx[1] \\
& \qw                  &    \qw & \gate{T} & \qw
}}}
=
\scalebox{0.91}{\raisebox{1.3cm}{
\Qcircuit @C=1.5em @R=2em  {
& \meter & \cw & \control \cw \cwx[2] & \cw \\
& \meter & \control \cw \cwx[1] \\
&    \qw & \gate{T} & \gate {T^{-1}} & \qw
}}}
\]
\caption{Trick 3: We can remove the controlled transformation $C_{X,1}^{-1}$ before the double meters by changing the measurement controlled gate.\label{Fig:Trick3}}
\end{figure}
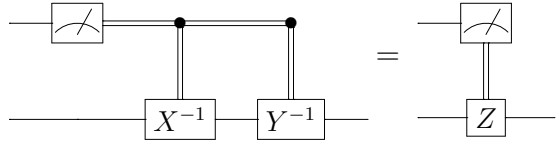
\begin{figure}[H]
\[
\scalebox{0.91}{\raisebox{.6cm}{
\Qcircuit @C=1.5em @R=2em  {
& \meter & \control \cw \cwx[1] & \control \cw \cwx[1]\\
&    \qw & \gate{X^{-1}} & \gate {Y^{-1}} & \qw
}}}
=
\scalebox{0.91}{\raisebox{.6cm}{
\Qcircuit @C=1.5em @R=2em  {
& \meter \cwx[1] \\
&  \gate{Z} & \qw
}}}
\]
\caption{Trick 4:
Since $\displaystyle Y^{-i}X^{-i}=\zeta^{-i^2} Z^i$, and the phase does not count in the protocol, we can simplify meter-controlled transformations.}\label{Fig:Trick4}
\end{figure}

\setcounter{equation}{0}
\section{Holographic Software \label{HoloSoft}}
In this section we give the dictionary to translate between pictorial protocols and the algebraic ones.  Any algebraic protocol can be translated into a pictorial protocol in a straightforward way.
From this picture we may be able to obtain new insights into the protocol.

We also give a dictionary for the inverse direction. Actually this is more interesting, as the pictures may be more intuitive: one says that 1 picture is worth 1,000 words.
In fact we give a new way to design protocols: we rely on the aesthetics of a picture as motivation for the structure of the protocol.  In this way, we can strive to introduce pictorial protocols which simulate human thought.

\subsection{Pictures for fundamental concepts}
Before we give the complete list of pictorial relations and our dictionary for translation, let us remark how some fundamental concepts in quantum information fit into our pictorial framework.
In \S\ref{Sect:PFA} we remark that one can write any
$n$-qudit transformation as an element in the parafermion algebra with $2n$ generators.
We represent the basis element $c_1^{k_1}c_2^{k_2}\cdots c_{2n}^{k_{2n}}$ in the parafermion algebra as a picture with $2n$ ``through'' strings\footnote{{By ``through'' string, we mean a neutral string that passes from the $j^{th}$ input to the $j^{th}$ output and that crosses no other string.}}, with the $j^{\rm th}$ string labelled by $k_j$ (on the left side).  The label is called the charge of the string, and the labels are positioned in an increasing vertical order:
\be
c_1^{k_1}c_2^{k_2}\cdots c_{2n}^{k_{2n}}
=
\raisebox{-.4cm}{
\tikz{
\draw (-1.5\mn,0)--(-1.5\mn,1);
\draw (-3\mn,0)--(-3\mn,1);
\draw (-6.5\mn,0)--(-6.5\mn,1);
\node at (-4\mn,1/2) {$\cdots$};
\node at (-1\mn,1/4) {\size{$k_1$}};
\node at (-2.5\mn,2/4) {\size{$k_2$}};
\node at (-5.5\mn,3/4) {\size{$k_{2n}$}};
}} \;.
\ee
The algebraic relations \eqref{ParafermionAlgebra} that permute the order of factors in the product of parafermions, become elementary relations between pictures, that we will give in \eqref{AddCharge}-\eqref{Equ:para isotopy}.  Besides these relations mentioned here, we give other pictorial relations in \S \ref{Sec:planar relation} and \S \ref{Sec:braided relations}. In addition, we give examples of how to apply these relations to quantum information.

We can also represent $n$-qudits as pictures. First we represent the $n$-qudit  zero-particle state $\vert\,\vec 0\,\rangle$, which up to a scalar is given by the  $n$-cap picture:
\be
d^{n/4}\,|\vec 0\rangle
\ =\
\raisebox{-.2cm}{\tikz{
\fqudit{0\mn}{0\nn}{1\mn}{1\nn}{}
\fqudit{4\mn}{0\nn}{1\mn}{1\nn}{}
\fqudit{-6\mn}{0\nn}{1\mn}{1\nn}{}
\node at (-4\mn,1\nn) {$\cdots$};
}}\quad.
\ee
The action of the parafermion algebra on the state $\vert\,\vec 0\,\rangle$ is captured by the joint relations between the charged strings and the caps given in Equation \eqref{Equ:SF1}.

It is extremely important that the multipartite entangled resource state $\Max$ (even in the case of multiple-persons) can be represented as the picture in Figure~\ref{Pic: Resource state}. This representation provides new insights in multipartite communication, which we explain later in this paper, and also in~\cite{ConstructiveSimulation}.

\subsection{Elementary notions}
We use a convention in identifying algebraic formulas with pictorial  ones: the objects on the left side of an equation are represented by the objects on the right side of the equation.

In our pictures, we call the points on top \textit{input points}, and the points on bottom \textit{output points}. The multiplication goes from bottom to top, and glues input points to output points. Tensor products go from left to right.

An $n$-qudit has 0 input points and $2n$ output points.
A dual $n$-qudit has $2n$ input points and 0 output points.
We call a picture with $n$ input points and $n$ output points an $n$-string transformation.
An $n$-qudit transformation is a $2n$-string transformation. (An $n$-qudit transformation is an $n$-string transformation in previous pictorial  approaches, such as in \cite{AC,KL-2}.) It is interesting that one can also talk about a $1$-string transformation that acts on ``$\frac{1}{2}$-qudits''. We refer the readers to \cite{ConstructiveSimulation} for an application of this concept.
We call
\be
\raisebox{-.4cm}{
\tikz{
\draw (0,0)--(0,1);
\node at (-1/6,1/2) {$k$};
}}
\ee
a charge-$k$ string, or a string with a  charge-$k$ particle. We write the label to the left of the string.

\subsection{Planar relations}\label{Sec:planar relation}
In this section we give relations between certain pictures.
The consistency of these relations is proved in \cite{JL}. Using these relations, we give a dictionary between qudits, transformations, measurements, and pictures.

\subsubsection{Addition of charge, and charge order}
\be\label{AddCharge}
\raisebox{-.5cm}{
\begin{tikzpicture}
\draw (0,0) --(0,1);
\node (0,0) at (-0.2,0.8) {$\ell$};
\node (0,0) at (-0.2,0.2) {$k$};
\node (0.45,0.25) at (0.45,0.5) {$=$};
\node (0.9,0.35) at (1.4,0.5) {$k + \ell$};
\draw (2,0) --(2,1);
\node (1.4,-0.05) at (2.3,0.3) {,};
\end{tikzpicture} }\qquad
\raisebox{-.5cm}{
\tikz{
\node (3.7,0.35) at (4.8,0.5) {$d$};
\draw (5,0) --(5,1);
\node (4.4,0.25) at (5.4,0.5) {$=$};
\draw (5.8,0) --(5.8,1);
\node (4.95,0) at (6.1,0.3) {.};
}}
\ee

\subsubsection{Para isotopy}\label{Sect:ParaIsotopy}
\be \label{Equ:para isotopy}
\raisebox{-.5cm}{
\begin{tikzpicture}
\draw (0,0) --(0,01);
\draw (0.2,0) --(0.2,01);
\draw [fill] (0.4,0.5) circle [radius=0.01];
\draw [fill] (0.5,0.5) circle [radius=0.01];
\draw [fill] (0.6,0.5) circle [radius=0.01];
\node (0,0) at (-0.15,0.2) {$k$};
\draw (0.8,0) --(0.8,1);
\draw (1.15,0) --(1.15,1);
\node (1.2,0) at (1,0.8) {$\ell$};
\end{tikzpicture}
}
=q^{k\ell}
\raisebox{-.4cm}{
\begin{tikzpicture}
\draw (0,0) --(0,01);
\draw (0.2,0) --(0.2,01);
\draw [fill] (0.4,0.5) circle [radius=0.01];
\draw [fill] (0.5,0.5) circle [radius=0.01];
\draw [fill] (0.6,0.5) circle [radius=0.01];
\node (0,0) at (-0.15,0.8) {$k$};
\draw (0.8,0) --(0.8,1);
\draw (1.15,0) --(1.15,1);
\node (1.2,0) at (1,0.2) {$\ell$};
\end{tikzpicture}}\;.
\ee
Here the strings between charge-$k$ string  and charge-$\ell$ string are not charged. We call $q^{k\ell}$ the twisting scalar.

\textbf{Notation:}
The twisted tensor product of pairs interpolates between the two vertical orders of the product.  In the twisted product, we write the labels at the same vertical height:
\beq\label{TwistedProduct}
\raisebox{-.4cm}{
\begin{tikzpicture}
\draw (0,0) --(0,01);
\draw (0.2,0) --(0.2,01);
\draw [fill] (0.4,0.5) circle [radius=0.01];
\draw [fill] (0.5,0.5) circle [radius=0.01];
\draw [fill] (0.6,0.5) circle [radius=0.01];
\node (0,0) at (-0.15,0.5) {$\scriptstyle{k}$};
\draw (0.8,0) --(0.8,1);
\draw (1.15,0) --(1.15,1);
\node (1.2,0) at (1,0.5) {$\scriptstyle{\ell}$};
\end{tikzpicture}}
&\equiv&\zeta^{-k\ell}
\raisebox{-.4cm}{
\begin{tikzpicture}
\draw (0,0) --(0,01);
\draw (0.2,0) --(0.2,01);
\draw [fill] (0.4,0.5) circle [radius=0.01];
\draw [fill] (0.5,0.5) circle [radius=0.01];
\draw [fill] (0.6,0.5) circle [radius=0.01];
\node (0,0) at (-0.15,0.2) {$\scriptstyle{k}$};
\draw (0.8,0) --(0.8,1);
\draw (1.15,0) --(1.15,1);
\node (1.2,0) at (1,0.8) {$\scriptstyle{\ell}$};
\end{tikzpicture}}
\nonumber\\
&=&\zeta^{k\ell}
\raisebox{-.4cm}{
\begin{tikzpicture}
\draw (0,0) --(0,01);
\draw (0.2,0) --(0.2,01);
\draw [fill] (0.4,0.5) circle [radius=0.01];
\draw [fill] (0.5,0.5) circle [radius=0.01];
\draw [fill] (0.6,0.5) circle [radius=0.01];
\node (0,0) at (-0.15,0.8) {$\scriptstyle{k}$};
\draw (0.8,0) --(0.8,1);
\draw (1.15,0) --(1.15,1);
\node (1.2,0) at (1,0.2) {$\scriptstyle{\ell}$};
\end{tikzpicture}
}\;.
\eeq

In this case $k,l\in \Z$, and $k$ and $k+d$ yield  different pictures. If the pair is neutral, namely $\ell=-k$, then the twisted tensor product is defined for $k\in \Z_d$. This twisted product was introduced in \cite{JP,JJ}.

\subsubsection{String Fourier relation}
\begin{align}
   \raisebox{-0.2cm}{
   \tikz{
   \fqudit {0}{0}{1\mn}{1\nn}{}
   \node at (-3\mn,1\nn) {\size{$\hskip -2.5cm k$}};
   }}
    &=\zeta^{k^2}~~
    \raisebox{-0.2cm}{\tikz{\fqudit {0}{0}{1\mn}{1\nn}{\phantom{kl} k}}}
    \ , \label{Equ:SF1}\\
    &\phantom{x}\nonumber\\
   \raisebox{-0.2cm}{
   \tikz{
   \fmeasure {0}{0}{1\mn}{1\nn}{};
   \node at (-3\mn,-1\nn) {$\hskip -2.5cm k$};
   }}
    &=\zeta^{-k^2}
    \raisebox{-0.2cm}{\tikz{\fmeasure {0}{0}{1\mn}{1\nn}{\phantom{kl} k}}}
    \ .\label{Equ:SF2}
\end{align}

\subsubsection{Quantum dimension}

\be\label{QuantumDimension}
\raisebox{-.4cm}{
\begin{tikzpicture}
\draw (0.8,0) circle [radius=0.5];
\end{tikzpicture}}
= \sqrt{d}
\;.
\ee

\subsubsection{Neutrality}
\be\label{Neutrality}
\raisebox{-.4cm}{
\begin{tikzpicture}
\node (0,0) at (0.45,0.05) {$k$};
\draw (1.1,0) circle [radius=0.5];
\end{tikzpicture}}
=0\;,\quad \text{for } d\nmid k.
\ee

\subsubsection{Temperley-Lieb relation}
\be
\raisebox{-.7cm}{
\tikz{
\draw (0,0)--(0,1/3+1/6) arc (180:0:1/6) arc (180:360:1/6) -- (2/3,1);
\draw (2/3+0,0+1/2)--(2/3+0,1/3+1/6+1/2) arc (0:180:1/6) arc (0:-180:1/6) -- (2/3+-2/3,1+1/2);
}}\
=
\raisebox{-.7cm}{
\tikz{
\draw (0,0)--(0,1+1/2);
}}\;,\qquad
\raisebox{-.7cm}{
\tikz{
\draw (0,0)--(0,1/3+1/6) arc (0:180:1/6) arc (0:-180:1/6) -- (-2/3,1);
\draw (-2/3,1/2)--(-2/3,1/3+1/6+1/2) arc (180:0:1/6) arc (180:360:1/6) -- (0,1+1/2);
}}\
=
\raisebox{-.7cm}{
\tikz{
\draw (0,0)--(0,1+1/2);
}}\; .
\ee

\textbf{Notation:}
Based on the Temperley-Lieb relation, a string only depends on the end points:
\be
\raisebox{-.5cm}{
\tikz{
\draw (0,0)--(0,1/3+1/6) arc (180:0:1/6) arc (180:360:1/6) -- (2/3,1);
}}\
=
\raisebox{-.5cm}{
\tikz{
\draw (0,0)--(2/3,1);
}}\;,\qquad
\raisebox{-.5cm}{
\tikz{
\draw (0,0)--(0,1/3+1/6) arc (0:180:1/6) arc (0:-180:1/6) -- (-2/3,1);
}}\
=
\raisebox{-.5cm}{
\tikz{
\draw (0,0)--(-2/3,1);
}}\; .
\ee

\subsubsection{Resolution of the identity}
\begin{align}\label{Equ:Resolution of the identity}
\raisebox{-.5cm}{
\tikz{
\draw (0,0)--(0,1+1/6);
\draw (2/3,0)--(2/3,1+1/6);
}}
&= d^{-1/2}\sum_{k=0}^{d-1}
\raisebox{-.5cm}{
\tikz{
\fqudit {0}{-3.5\nn}{1\mn}{.5\nn}{\phantom{-}k}
\fmeasure {0}{0}{1\mn}{.5\nn}{-k}
}}\quad .
\end{align}

\subsection{$1$-Qudit dictionary}
Now we give the first pictorial translations of the algebraic formulas.  It will be evident from the context of the picture, when a symbol such as  $k$ denotes a label, in contrast with $d^{-1/4}$ or $q^{k\ell}$ or $\zeta^{k^{2}}$, that denote a scalar multiple.

\subsubsection{Qudit}
Our picture for the qudit $\ket{k}$ is:
\begin{equation}\label{qudit 21}
\ket{k}=d^{-1/4}
\raisebox{-.2cm}{
\tikz{
   \fqudit {0}{0}{1\mn}{1\nn}{\phantom{-}k}
}}\ .
\end{equation}
We place the label on the righthand string in the cap by convention.

\subsubsection{Dual qudit}
The picture for the dual-qudit  $\bra{k}$  is:
\be
\label{dualqudit 21}
\bra{k}=d^{-1/4}
\raisebox{-.2cm}{
\tikz{
\fmeasure {0}{0}{1\mn}{1\nn}{\,-k}
}}\;.
\ee

\subsubsection{Transformations}
 Transformations $T$ of $1$-qudits are pictures with two input points and two output points.  The identity transformation is
\[
I =
\raisebox{-.47cm}{
\tikz{
\draw (0,0)--(0,1);
\draw (2/3,0)--(2/3,1);
}}\;.
\]

\subsubsection{Matrix Units}
The picture for the transformation $\ket{k}\bra{\ell}$ is
\be
\label{matrix units}
\ket{k}\bra{\ell}=d^{-1/2}
\raisebox{-.5cm}{
\tikz{
\fqudit {0}{-3.5\nn}{1\mn}{.5\nn}{\phantom{kk} k}
\fmeasure {0}{0}{1\mn}{.5\nn}{\ -\ell}
}}\;.
\ee

\subsubsection{Pauli matrices $X,Y,Z$ \label{Sect:PauliXYZ}}
The pictures for  Pauli $X, Y, Z$  are:
\be
\label{Pauli 21}
X=
\raisebox{-.4cm}{
\tikz{
\draw (0,0)--(0,1);
\draw (2/3,0)--(2/3,1);
\node at (-1/3+2/3,1/2) {1};
}}\;,\quad
Y=
\raisebox{-.4cm}{
\tikz{
\draw (0,0)--(0,1);
\draw (2/3,0)--(2/3,1);
\node at (-1/3,1/2) {-1};
}}\;,\quad
Z=
\raisebox{-.4cm}{
\tikz{
\draw (0,0)--(0,1);
\draw (2/3,0)--(2/3,1);
\node at (-1/3,1/2) {1};
\node at (-1/3+2/3,1/2) {-1};
}}\;.
\ee

\subsection{$1$-Qudit properties}
In this section we explain why the dictionary is holographic for $1$-qudits, and we show how the Pauli $X,Y,Z$ in \eqref{Pauli 21} actually correspond to the usual qudit Pauli matrices.

\begin{itemize}
\item{} Orthonormal Basis:
\be
\bra{\ell}k\rangle
=
d^{-1/2}\, \raisebox{-.6cm}{
\tikz{
\fqudit {0}{0\nn}{1\mn}{1\nn}{\phantom{kl} k}
\fmeasure {0}{0}{1\mn}{1\nn}{\,-\ell}
}}
= \delta_{\ell, k}\;.
\ee
Here we use the relations \eqref{AddCharge}, \eqref{QuantumDimension}, and \eqref{Neutrality}.

\item{} Transformations:
The matrix units  $\ket{k}\bra{\ell}$ are pictured in  \eqref{matrix units}.
\end{itemize}

Therefore single qudit transformations can be represented as pictures.
On the other hand, relation (\ref{Equ:Resolution of the identity}) indicates that any picture with two input points and two output points is a single qudit transformation. This gives an elementary dictionary for translation between single qudit transformations and pictures with two input points and two output points. In general, there is a correspondence between $n$-qudit transformations and pictures with $2n$ input points and $2n$ output points.

In this way, the pictorial computation is the same as the usual algebraic computation in quantum information.

\begin{itemize}
\item{Pauli $X,Y,Z$ Relations:}
Using the notation for qudits in \eqref{qudit 21} and \eqref{Pauli 21}, one can identify these three $2$-string transformations as the Pauli matrices defined in \eqref{XYZ-Defn-1}.
\beq\label{PauliReduce1}
\raisebox{-.25cm}{
\tikz{
\draw (0,1/3)--(0,1);
\draw (2/3,1/3)--(2/3,1);
\node at (-1/3+2/3,1/2) {\phantom{k} 1};
\draw (0,1) arc (180:0:1/3);
\node at (-1/3+2/3,1) {\phantom{k} $k$};
}}
&=&
\raisebox{-.25cm}{
\tikz{
\draw (0,1/3)--(0,1);
\draw (2/3,1/3)--(2/3,1);
\draw (0,1) arc (180:0:1/3);
\node at (-1/6+1/2,1) {\size{$k$+1}};
}}\;,\\
\label{PauliReduce2}
\raisebox{-.25cm}{
\tikz{
\draw (0,1/3)--(0,1);
\draw (2/3,1/3)--(2/3,1);
\draw (0,1) arc (180:0:1/3);
\node at (-1/6+2/3,1) {$k$};
\node at (-1/3,1/2) {\,-1};
}}
&=&
\zeta^{1-2k}
\raisebox{-.25cm}{
\tikz{
\draw (0,1/3)--(0,1);
\draw (2/3,1/3)--(2/3,1);
\draw (0,1) arc (180:0:1/3);
\node at (-1/8+1/2,1) {\size{$k$-1}};
}}\;,\\
\label{PauliReduce3}
\raisebox{-.25cm}{
\tikz{
\draw (0,1/3)--(0,1);
\draw (2/3,1/3)--(2/3,1);
\node at (-1/6,1/2) {1};
\draw (0,1) arc (180:0:1/3);
\node at (1/3,1) {\phantom{k} $k$};
\node at (1/3,1/2) {\phantom{l} -1};
}}
&=&q^k
\raisebox{-.25cm}{
\tikz{
\draw (0,1/3)--(0,1);
\draw (2/3,1/3)--(2/3,1);
\draw (0,1) arc (180:0:1/3);
\node at (1/3,1) {\phantom{l} $k$};
}}\;.
\eeq
The pictorial equalities in \eqref{PauliReduce1}--\eqref{PauliReduce3} are a consequence of the relations \eqref{AddCharge}--\eqref{Equ:SF1}.

\item{} Vertical reflection or adjoint:
The vertical reflection of pictures maps the particle of charge $k$ to the particle of charge $-k$.   This involution is an anti-linear, anti-isomorphism of pictures.  It interchanges  $\ket{k}$ with $\bra{k}$.  For qudits or transformations, the vertical reflection is the usual  adjoint $^{*}$.
\end{itemize}

\subsection{$n$-Qudit dictionary}
We mainly discuss the $2$-qudit case. One can easily generalize the argument to the case of $n$-qudits.

\subsubsection{Elementary dictionary}
There are two different ways to represent 2-qudits as pictures indicated by the arrow.
\be\label{Down-Basis}
\ket{k_1,k_2}^\searrow
=\frac{1}{d^{1/2}}\ \raisebox{-.2cm}{\tikz{
\fqudit{0\mn}{0\nn}{1\mn}{1\nn}{\phantom{ll} k_2}
\fqudit{4\mn}{0\nn}{1\mn}{2\nn}{\phantom{ll} k_1}
}} \;,
\ee
or
\be
\ket{k_1,k_2}^\nearrow
=\frac{1}{d^{1/2}}\ \raisebox{-.2cm}{\tikz{
\fqudit{0\mn}{0\nn}{1\mn}{2\nn}{\phantom{ll} k_2}
\fqudit{4\mn}{0\nn}{1\mn}{1\nn}{\phantom{ll} k_1}
}} \;. \label{Equ:increasing basis}
\ee
Then two representations give two different dictionaries, but they are unitary equivalent.
We fix the first choice $\ket{k_1,k_2}=\ket{k_1,k_2}^\searrow$ in our software, since it works out better with concepts in quantum information.

We represent an $n$-qudit $\vec{\ket{k}}=\ket{k_1,k_2,\cdots, k_n}$ as
\be\label{DecreasingBasis}
\vec{\ket{k}}
= \vec{\ket{k}}^\searrow
=\frac{1}{d^{n/4}}\ \raisebox{-.5cm}{\tikz{
\fqudit{0\mn}{0\nn}{1\mn}{2\nn}{\phantom{k} k_2}
\fqudit{4\mn}{0\nn}{1\mn}{3\nn}{\phantom{k} k_1}
\fqudit{-6\mn}{0\nn}{1\mn}{1\nn}{\phantom{k} k_n}
\node at (-4\mn,1\nn) {$\cdots$};
}}\quad.
\ee

We can represent the $n$-qudit transformation 
\[
\vec{\ket{k}}\vec{\bra{\ell}}=\ket{k_1,k_2,\cdots, k_n}\bra{\ell_1,\ell_2,\cdots, \ell_n}
\] 
as
\be
\vec{\ket{k}}\vec{\bra{\ell}}
=\frac{1}{d^{n/2}}
\raisebox{-1.2cm}{
\tikz{
\fqudit{0\mn}{0\nn}{1\mn}{2\nn}{\phantom{k} k_2}
\fqudit{4\mn}{0\nn}{1\mn}{3\nn}{\phantom{k} k_1}
\fqudit{-6\mn}{0\nn}{1\mn}{1\nn}{\phantom{k} k_n}
\node at (-4\mn,7\nn) {$\cdots$};
\fmeasure{0\mn}{9\nn}{1\mn}{2\nn}{-\ell_2}
\fmeasure{4\mn}{9\nn}{1\mn}{3\nn}{-\ell_1}
\fmeasure{-6\mn}{9\nn}{1\mn}{1\nn}{-\ell_n}
\node at (-4\mn,1\nn) {$\cdots$};
}}\quad.
\ee

We denote an $n$-qudit transformation $T$ as

\be
T=\underbrace{\raisebox{-.4cm}{
\tikz{
\draw (-1/6,1/3) rectangle (1/6+2/3,1-1/3);
\draw (0,0)--(0,1/3);
\draw (2/3,0)--(2/3,1/3);
\draw (0,1)--(0,1-1/3);
\draw (2/3,1)--(2/3,1-1/3);
\node at (1/3,1/2) {$T$};
\node at (1/3,1/6) {$\cdots$};
\node at (1/3,1-1/6) {$\cdots$};
}}
}_{2n} \;.
\ee

In the other direction, any picture with 0 input points and $2n$-outpoint points is an $n$-qudit.
Any picture with $2n$ input points and $2n$ outpoint points is an $n$-qudit transformation.

\subsubsection{Controlled transformations}
Suppose $T$
is a single qudit transformation.  Now we give the pictorial representation of the controlled transformations $C_{1,T}$ and $C_{T,1}$ in Figs.~\ref{Fig:Controlled gate 1} and \ref{Fig:Controlled gate 2}. 
\begin{align}
C_{1,T}&=\frac{1}{\sqrt{d}}\sum_{k=0}^{d-1}
~\raisebox{-.8cm}{
\tikz{
\fqudit {-3/3}{-1/2}{-1/3}{5/6}{k};
\fmeasure {-3/3}{1+1/3}{-1/3}{1/6}{-k};
\draw (0,0-1/2)--(0,4/3);
\draw (2/3,0-1/2)--(2/3,4/3);
\fill[white] (-1/6,-1/3) rectangle (1/6+2/3,1-4/6);
\draw (-1/6,-1/3) rectangle (1/6+2/3,1-4/6);
\node at (1/3,0) {$T^k$};
}}\;, \\
\phantom {a}\nonumber\\
C_{T,1}&=\frac{1}{\sqrt{d}}\sum_{k=0}^{d-1}
\raisebox{-.8cm}{
\tikz{
\fqudit {3/3}{-1/3}{-1/3}{1/3}{\phantom{kl} k};
\fmeasure {3/3}{1+1/3}{-1/3}{1/3}{\,-k};
\draw (-1/6,1/6) rectangle (1/6+2/3,1-1/6);
\draw (0,0-1/3)--(0,1/6);
\draw (2/3,0-1/3)--(2/3,1/6);
\draw (0,1+1/3)--(0,1-1/6);
\draw (2/3,1+1/3)--(2/3,1-1/6);
\node at (1/3,1/2) {$T^k$};
}}~\;.
\end{align}
In particular, $C_{Z}\equiv C_{Z,1}=C_{1,Z}$. 

\begin{lemma}\label{Lem:CZ}
Note that $C_{Z}\ket{k_1,k_2}=q^{k_1k_2}\ket{k_1,k_2}.$ We have that

\be
\raisebox{-.2cm}{\tikz{
\fsqudit{0/3}{0\nn}{1/3}{2\nn}{\phantom{kl}k_1}
\fsqudit{4/3}{0\nn}{1/3}{1\nn}{\phantom{kl}k_2}
\draw (-2.5/3,0) rectangle (4.5/3,-2\nn);
\node at (1/3,-1\nn) {\size{$C_{Z}$}};
\draw (-2/3,-2\nn)--(-2/3,-2.5\nn);
\draw (0/3,-2\nn)--(0/3,-2.5\nn);
\draw (2/3,-2\nn)--(2/3,-2.5\nn);
\draw (4/3,-2\nn)--(4/3,-2.5\nn);
}}
=
\raisebox{-.2cm}{\tikz{
\fsqudit{0/3}{0\nn}{1/3}{3.5\nn}{\phantom{kl}k_1}
\fsqudit{4/3}{0\nn}{1/3}{4.5\nn}{\phantom{kl}k_2}
}}\quad.
\ee
\end{lemma}
\begin{proof}
The equation follows from para isotopy in Equation \eqref{Equ:para isotopy}.
\end{proof}

\subsubsection{$1$-Qudit transformations  on $2$-qudits}
A $1$-qudit transformation $T$ can act on $2$-qudits by adding two strings on the left or on the right.
We can translate these pictorial transformations to algebraic ones as follows:

\begin{lemma}\label{Lem:T}
For a 1-qudit transformation $T$, we have that
\begin{align}
\raisebox{-.3cm}{
\tikz{
\draw (-1/6,1/6) rectangle (1/6+2/3,1-1/6);
\draw (0,0)--(0,1/6);
\draw (2/3,0)--(2/3,1/6);
\draw (0,1)--(0,1-1/6);
\draw (2/3,1)--(2/3,1-1/6);
\node at (1/3,1/2) {$T$};
\draw (-2/3,0)--(-2/3,1);
\draw (-4/3,0)--(-4/3,1);
}}
=&1 \otimes T \;,\\
& \nonumber\\
\raisebox{-.3cm}{
\tikz{
\draw (-1/6,1/6) rectangle (1/6+2/3,1-1/6);
\draw (0,0)--(0,1/6);
\draw (2/3,0)--(2/3,1/6);
\draw (0,1)--(0,1-1/6);
\draw (2/3,1)--(2/3,1-1/6);
\node at (1/3,1/2) {$T$};
\draw (4/3,0)--(4/3,1);
\draw (6/3,0)--(6/3,1);
}}~\;=&C_{Z} (T\otimes 1) C_{Z}^{-1}\;. \label{Equ:1T1}
\end{align}
Furthermore, if $T$ has charge $k$, then this action equals
\be
\raisebox{-.3cm}{
\tikz{
\draw (-1/6,1/6) rectangle (1/6+2/3,1-1/6);
\draw (0,0)--(0,1/6);
\draw (2/3,0)--(2/3,1/6);
\draw (0,1)--(0,1-1/6);
\draw (2/3,1)--(2/3,1-1/6);
\node at (1/3,1/2) {$T$};
\draw (4/3,0)--(4/3,1);
\draw (6/3,0)--(6/3,1);
}}
=T \otimes Z^{k}
\;. \label{Equ:1T2}
\ee

\end{lemma}

\begin{proof}
The first equation follows from the definition.
The second equation follows from Lemma \ref{Lem:CZ}.
The third equation follows from the algebraic identity $T \otimes Z^{k}=C_{Z} (T\otimes 1) C_{Z}^{-1}$, for a charged $k$ transformation $T$.
\end{proof}

Note using \eqref{Equ:1T2} is preferable to using  \eqref{Equ:1T1}, since $Z^{k}$ and $T$ can be performed locally by two persons.

\begin{proposition}\label{Jordan-Wigner Form}
If $T$  is a charge $k$, multi-qudit transformation, then
\begin{align}\label{Jordan-WignerForm}
&\raisebox{-1.1cm}{
\tikz{
\draw (-1/6,1/3) rectangle (1/6+2/3,1-1/3);
\draw (0,0)--(0,1/3);
\draw (2/3,0)--(2/3,1/3);
\draw (0,1)--(0,1-1/3);
\draw (2/3,1)--(2/3,1-1/3);
\node at (1/3,1/2) {$T$};
\node at (1/3,1/6) {$\cdots$};
\node at (1/3,1-1/6) {$\cdots$};
\node at (5/3,5/24) {$\underbrace{
\tikz{
\draw (-2/3,0)--(-2/3,1);
\draw (-3/3,0)--(-3/3,1);
\draw (-5/3,0)--(-5/3,1);
\draw (-6/3,0)--(-6/3,1);
\node at (-4/3,1/2) {$\cdots$};
}
}_{2n}$};
\node at (-3/3,5/24) {$\underbrace{
\tikz{
\draw (-2/3,0)--(-2/3,1);
\draw (-3/3,0)--(-3/3,1);
\draw (-5/3,0)--(-5/3,1);
\draw (-6/3,0)--(-6/3,1);
\node at (-4/3,1/2) {$\cdots$};
}
}_{2m}$};
}} \nonumber \\
=&\underbrace{1\otimes \cdots \otimes 1}_{2m} \otimes T \otimes  \underbrace{Z^{k} \otimes \cdots \otimes Z^{k}}_{2n} \;.
\end{align}
\end{proposition}

\begin{proof}
Repeating Equation \eqref{Equ:1T2} in Lemma \ref{Lem:T}, we obtain the equality.
\end{proof}

\subsubsection{Jordan-Wigner transformations}
The Jordan-Wigner transformation is an isomorphism between the parafermion algebra with $2n$ generators and $n$-qudit transformations, namely  the $n$-fold tensor product of the $d\times d$ matrix algebra $M_{d}(\mathbb{C})$.
As a particular case of Proposition \ref{Jordan-Wigner Form}, we obtain the qudit Jordan-Wigner transformation for $T=X$, $Y$, or $Z$. 
We give an intuitive pictorial interpretation of this transformation in Proposition \ref{Prop:JWasLocal}.

\begin{proposition}
The Jordan-Wigner transformation is given by
\begin{align}\label{Equ:J-W x}
\raisebox{-.4cm}{
\tikz{
\draw (0\mn,0)--(0,1);
\draw (-1\mn,0)--(-1\mn,1);
\node at (1/6,1/2) {\rm{\,1}};
\draw (-2\mn,0)--(-2\mn,1);
\draw (-3\mn,0)--(-3\mn,1);
\draw (-5\mn,0)--(-5\mn,1);
\draw (-6\mn,0)--(-6\mn,1);
\node at (-4\mn,1/2) {$\cdots$};
}}
&=
X \otimes Z \otimes \cdots \otimes Z
 \;,\\ \label{Equ:J-W y}
\raisebox{-.4cm}{
\tikz{
\draw (0,0)--(0,1);
\draw (-1\mn,0)--(-1\mn,1);
\node at (-1/6,1/2) {\rm-1};
\draw (-2\mn,0)--(-2\mn,1);
\draw (-3\mn,0)--(-3\mn,1);
\draw (-5\mn,0)--(-5\mn,1);
\draw (-6\mn,0)--(-6\mn,1);
\node at (-4\mn,1/2) {$\cdots$};
}}
&=
Y \otimes Z^{-1} \otimes \cdots \otimes Z^{-1}
\;,\\ \label{Equ:J-W z}
\raisebox{-.4cm}{
\tikz{
\draw (0,0)--(0,1);
\draw (-1\mn,0)--(-1\mn,1);
\node at (1/6,1/2) {{\rm-1}};
\node at (-1/6,1/2) {\rm1};
\draw (-2\mn,0)--(-2\mn,1);
\draw (-3\mn,0)--(-3\mn,1);
\draw (-5\mn,0)--(-5\mn,1);
\draw (-6\mn,0)--(-6\mn,1);
\node at (-4\mn,1/2) {$\cdots$};
}}
&=
Z \otimes 1 \otimes \cdots \otimes 1
\;.
\end{align}
Equivalently, we can represent Pauli matrices on $n$-qudits as diagrams.
\begin{align}
X \otimes  1 \otimes \cdots \otimes 1
&= ~~~\raisebox{-.4cm}{
\tikz{
\draw (0,0)--(0,1);
\draw (-1\mn,0)--(-1\mn,1);
\node at (1/6,1/2) {\rm 1};
\draw (-2\mn,0)--(-2\mn,1);
\draw (-3\mn,0)--(-3\mn,1);
\draw (-6\mn,0)--(-6\mn,1);
\draw (-7\mn,0)--(-7\mn,1);
\node at (-4\mn,1/2) {$\cdots$};
\node at (-1.5\mn,1/2) {\rm -1};
\node at (-2.5\mn,1/2) {\rm 1};
\node at (-5.5\mn,1/2) {\rm -1};
\node at (-6.5\mn,1/2) {\rm 1};
}} \;,\\
Y \otimes  1 \otimes \cdots \otimes 1
&= \raisebox{-.4cm}{
\tikz{
\draw (0,0)--(0,1);
\draw (-1\mn,0)--(-1\mn,1);
\node at (-1/6,1/2) {\rm -1};
\draw (-2\mn,0)--(-2\mn,1);
\draw (-3\mn,0)--(-3\mn,1);
\draw (-6\mn,0)--(-6\mn,1);
\draw (-7\mn,0)--(-7\mn,1);
\node at (-4\mn,1/2) {$\cdots$};
\node at (-1.5\mn,1/2) {\rm 1};
\node at (-2.5\mn,1/2) {\rm -1};
\node at (-5.5\mn,1/2) {\rm 1};
\node at (-6.5\mn,1/2) {\rm -1};
}} \;,\\
Z \otimes 1 \otimes \cdots \otimes 1
&= \; \raisebox{-.4cm}{
\tikz{
\draw (0,0)--(0,1);
\draw (-1\mn,0)--(-1\mn,1);
\node at (1/6,1/2) {\rm -1};
\node at (-1/6,1/2) {\rm 1};
\draw (-2\mn,0)--(-2\mn,1);
\draw (-3\mn,0)--(-3\mn,1);
\draw (-6\mn,0)--(-6\mn,1);
\draw (-7\mn,0)--(-7\mn,1);
\node at (-4\mn,1/2) {$\cdots$};
}} \;.
\end{align}
\end{proposition}

\begin{proof}
This follows from Proposition \ref{Jordan-Wigner Form}. 
\end{proof}

\begin{remark}
If we work on the increasing basis in Equation \eqref{Equ:increasing basis}, then we obtain the following Jordan-Wigner transformation:
\begin{align}
\raisebox{-.4cm}{
\tikz{
\draw (0,0)--(0,1);
\draw (1/3,0)--(1/3,1);
\node at (1/6,1/2) {\rm \,1};
\draw (-2/3,0)--(-2/3,1);
\draw (-3/3,0)--(-3/3,1);
\draw (-5/3,0)--(-5/3,1);
\draw (-6/3,0)--(-6/3,1);
\node at (-4/3,1/2) {$\cdots$};
}}\;
&=
Z^{-1} \otimes \cdots \otimes Z^{-1} \otimes X
\;,\\
\raisebox{-.4cm}{
\tikz{
\draw (0,0)--(0,1);
\draw (1/3,0)--(1/3,1);
\node at (-1/6,1/2) {\rm -1};
\draw (-2/3,0)--(-2/3,1);
\draw (-3/3,0)--(-3/3,1);
\draw (-5/3,0)--(-5/3,1);
\draw (-6/3,0)--(-6/3,1);
\node at (-4/3,1/2) {$\cdots$};
}}
\; &=
 Z \otimes \cdots \otimes Z \otimes Y
 \;,\\
\raisebox{-.4cm}{
\tikz{
\draw (0,0)--(0,1);
\draw (1/3,0)--(1/3,1);
\node at (1/6,1/2) {\rm -1};
\node at (-1/6,1/2) {\rm 1};
\draw (-2/3,0)--(-2/3,1);
\draw (-3/3,0)--(-3/3,1);
\draw (-5/3,0)--(-5/3,1);
\draw (-6/3,0)--(-6/3,1);
\node at (-4/3,1/2) {$\cdots$};
}}
&=
1 \otimes \cdots \otimes 1 \otimes Z
\;.
\end{align}
Equivalently,
\begin{align}
1 \otimes \cdots \otimes 1 \otimes X
=& \raisebox{-.4cm}{
\tikz{
\draw (0,0)--(0,1);
\draw (1/3,0)--(1/3,1);
\node at (1/6,1/2) {\;\rm 1};
\draw (-2/3,0)--(-2/3,1);
\draw (-3/3,0)--(-3/3,1);
\draw (-6/3,0)--(-6/3,1);
\draw (-7/3,0)--(-7/3,1);
\node at (-4.8/3,1/2) {$\cdots$};
\node at (-5/6,1/2) {\,\rm -1};
\node at (-7/6,1/2) {\;\rm 1};
\node at (-13/6,1/2) {\,\rm -1};
\node at (-15/6,1/2) {\;\rm 1};
}} \;,\\
1 \otimes \cdots \otimes 1 \otimes Y
=& \raisebox{-.4cm}{
\tikz{
\draw (0,0)--(0,1);
\draw (1/3,0)--(1/3,1);
\node at (-1/6,1/2) {\rm -1};
\draw (-2/3,0)--(-2/3,1);
\draw (-3/3,0)--(-3/3,1);
\draw (-6/3,0)--(-6/3,1);
\draw (-7/3,0)--(-7/3,1);
\node at (-4.9/3,1/2) {$\cdots$};
\node at (-5/6,1/2) {\;\rm 1};
\node at (-7/6,1/2) {\rm -1};
\node at (-13/6,1/2) {\;\rm 1};
\node at (-15/6,1/2) {\rm -1};
}} \;,\\
1 \otimes \cdots \otimes 1 \otimes Z
=&\quad  \raisebox{-.4cm}{
\tikz{
\draw (0,0)--(0,1);
\draw (1/3,0)--(1/3,1);
\node at (1/6,1/2) {\,\rm -1};
\node at (-1/6,1/2) {\rm 1};
\draw (-2/3,0)--(-2/3,1);
\draw (-3/3,0)--(-3/3,1);
\draw (-6/3,0)--(-6/3,1);
\draw (-7/3,0)--(-7/3,1);
\node at (-4.6/3,1/2) {$\cdots$};
}} \;.
\end{align}
\end{remark}

\subsubsection{Measurement dictionary I}
When a protocol has a meter, and the measurement of this meter is $\ell$, it is the same as applying the dual qudit $\bra{\ell}$ to the corresponding qudit.

\begin{proposition}
If the measurement of a meter $m_{j}$ on the $j^{\rm th}$ qudit of an $n$-qudit is $\ell$, then the diagram is
\be
m_j=l \longrightarrow
\raisebox{-1.1cm}{
\tikz{
\fmeasure {2\mn}{1}{1\mn}{1\nn}{\text{-}\ell}
\node at (6\mn,5/24) {$\underbrace{
\tikz{
\draw (-2/3,0)--(-2/3,1);
\draw (-3/3,0)--(-3/3,1);
\draw (-6/3,0)--(-6/3,1);
\draw (-7/3,0)--(-7/3,1);
\node at (-4.6/3,1/2) {$\cdots$};
}
}_{2(j-1)}$};
\node at (-4\mn,5/24) {$\underbrace{
\tikz{
\draw (-2/3,0)--(-2/3,1);
\draw (-3/3,0)--(-3/3,1);
\draw (-6/3,0)--(-6/3,1);
\draw (-7/3,0)--(-7/3,1);
\node at (-4.6/3,1/2) {$\cdots$};
\node at (-2.5/3,1/2) {-$\ell$};
\node at (-3.5/3,1/2) {$\ell$};
\node at (-6.5/3,1/2) {-$\ell$};
\node at (-7.5/3,1/2) {$\ell$};
}
}_{2(n-j-1)}$};
}}\;.
\ee\end{proposition}

\begin{proof}
It follows from the para isotopy in Equation \eqref{Equ:para isotopy}.
\end{proof}

Conversely, the diagram
\be
\raisebox{-1.1cm}{
\tikz{
\fmeasure {2\mn}{1}{1\mn}{1\nn}{\text{-}\ell}
\node at (6\mn,5/24) {$\underbrace{
\tikz{
\draw (-2/3,0)--(-2/3,1);
\draw (-3/3,0)--(-3/3,1);
\draw (-6/3,0)--(-6/3,1);
\draw (-7/3,0)--(-7/3,1);
\node at (-4.6/3,1/2) {$\cdots$};
}
}_{2(j-1)}$};
\node at (-4\mn,5/24) {$\underbrace{
\tikz{
\draw (-2/3,0)--(-2/3,1);
\draw (-3/3,0)--(-3/3,1);
\draw (-6/3,0)--(-6/3,1);
\draw (-7/3,0)--(-7/3,1);
\node at (-4.5/3,1/2) {$\cdots$};
}
}_{2(n-j-1)}$};
}}
\ee

\noindent means that there is a meter on this $j^{\rm th}$ qudit of an $n$-qudit, and the measurement is $\ell$. Moreover, the result is sent to persons who possess the last $(n-j-1)$ qudits. Then the persons apply $Z^{-\ell}$ to each of the $(n-j-1)$ target qudits. The corresponding protocol is in Figure~\ref{Fig:dual qudit}.
Of course we can not predict the result of the measurement, so the pictorial protocol must work for all $\ell$.

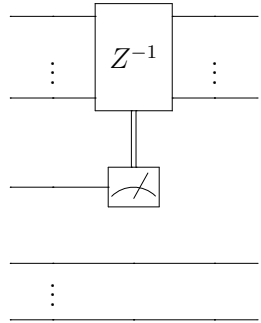
\begin{figure}[h]
\[
\scalebox{.91}{
\Qcircuit @C=1.5em @R=2em  {
& \qw &\multigate{1}{Z^{-1}} & \qw & \qw \\
& \ustick{ \vdots} \qw & \ghost{Z^{-1}} & \ustick{ \vdots} \qw & \qw \\
& \qw & \meter \cwx[-1]  \\
& \qw & \qw & \qw & \qw \\
& \ustick{ \vdots} \qw & \qw & \qw & \qw
}}
\]
\caption{Measurement controlled-$Z^{{-1}}$ gate: If the result of the measurement is $\ell$, then one applies  $Z^{-\ell}$ to the target qudits.}\label{Fig:dual qudit}
\end{figure}

\subsection{Braided relations}\label{Sec:braided relations}
\subsubsection{The braid\label{Sect:Braid}}
We begin by defining a positive and negative braid in terms of planar diagrams.   The braid acts on two strings, and although it is represented in the plane it satisfies the three Reidermeister moves characteristic of a braid.  These relations allow one to lift the planar relations to three-dimensional ones.  We refer the readers to \S8 of \cite{JL} for the proof of the braided relations that we state in this section.

Define $\displaystyle \omega = \frac{1}{\sqrt{d}} \sum_{j=0}^{d-1} \zeta^{j^{2}}$.  Then $|\omega|=1$, as shown in Proposition 2.15 of \cite{JL}.
Let $\omega^{1/2}$ be a fixed square root of $\omega$, and define the positive braid $b_{+}$ as
\beq\label{first-pos-braid}
b_{+}
&=\raisebox{-.3cm}{
\tikz{
\fbraid{0}{2\nn}{2\mn}{0}
}}
&\ \equiv
\frac{1}{\sqrt{\omega d}}\,
 \sum_{k=0}^{d-1}
\raisebox{-.3cm}{
\tikz{
\draw (0\mn,0) --(0\mn,3\nn);
\draw (2\mn,0) --(2\mn,3\nn);
\node at (1\mn,1\nn) {$-k$};
\node at (3\mn,2\nn) {$k$};
}}\\
&=&
\frac{1}{\sqrt{\omega d}}\,
\sum_{k=0}^{d-1} \zeta^{k^{2}}\,
\raisebox{-.3cm}{
\tikz{
\draw (0\mn,0) --(0\mn,3\nn);
\draw (2\mn,0) --(2\mn,3\nn);
\node at (1\mn,1.5\nn) {$-k$};
\node at (3\mn,1.5\nn) {$k$};
}}\quad.\nonumber
\eeq
Here we give two different expressions for the braid.  The  second formula involves the twisted product given in \eqref{TwistedProduct}.

The braid is a unitary gate. Its adjoint $b_{+}^*$ equals the inverse braid, the negative braid $b_{+}^{-1}=b_{-}$. In diagrams,
\beq\label{first-neg-braid}
 b_{+}^{*} = b_{-}
&=&\raisebox{-.3cm}{
\tikz{
\fbraid{2\mn}{2\nn}{0\mn}{0}
}}
=
\frac{\sqrt\omega }{\sqrt{d}}\,
 \sum_{k=0}^{d-1}
\raisebox{-.3cm}{
\tikz{
\draw (0\mn,0) --(0\mn,3\nn);
\draw (2\mn,0) --(2\mn,3\nn);
\node at (1\mn,2\nn) {$-k$};
\node at (3\mn,1\nn) {$k$};
}}\\
&=&
\frac{\sqrt\omega }{\sqrt{d}}\,
\sum_{k=0}^{d-1} \zeta^{-k^{2}}\,
\raisebox{-.3cm}{
\tikz{
\draw (0\mn,0) --(0\mn,3\nn);
\draw (2\mn,0) --(2\mn,3\nn);
\node at (1\mn,1.5\nn) {$-k$};
\node at (3\mn,1.5\nn) {$k$};
}}\quad.\nonumber
\eeq

The positive and negative braids \eqref{first-pos-braid}--\eqref{first-neg-braid} also have the equivalent definitions:
\beq
\raisebox{-.3cm}{
\tikz{
\fbraid{2\mn}{0}{0\mn}{2\nn}
}}
&=&
\frac{\sqrt\omega }{\sqrt{d}}\,
\sum_{k=0}^{d-1} \zeta^{-k^{2}}\,
\raisebox{-.6cm}{
\tikz{
\fqudit {0}{0\nn}{1\mn}{.5\nn}{\phantom{ll}k}
\fmeasure {0}{3.5\nn}{1\mn}{.5\nn}{-k}
}}\;,
\eeq
\be\label{second-neg-braid}
\raisebox{-.3cm}{
\tikz{
\fbraid{0\mn}{0}{2\mn}{2\nn}
}}
=
\frac{1}{\sqrt{\omega d}}\,
\sum_{k=0}^{d-1} \zeta^{k^{2}}\,
\raisebox{-.6cm}{
\tikz{
\fqudit {0}{0\nn}{1\mn}{.5\nn}{\phantom{ll}k}
\fmeasure {0}{3.5\nn}{1\mn}{.5\nn}{-k}
}} \;.
\ee

\bigskip
These definitions lead to the following braided relations:
\subsubsection{Braid-Fourier relation}
\be
\raisebox{-.47cm}{
\tikz{
\draw (0,0)--(1/3,1/3) arc (135:-45:1.414/6) -- (1/3,-1/3);
\draw (0,-1/3)--(1/6,-1/6) arc (-45:15:1.414/6);
\draw (0,0) arc (-45:135-360:1.414/6) --(0,2/3);
\draw (1/3,2/3)--(1/6,1/2) arc (135:195:1.414/6);
}}
=
\raisebox{-.3cm}{
\tikz{
\draw (1.5,6.5) --(1.75,6.75);
\draw (1,6.75) --(1.75,6);
\draw (1,6) --(1.25,6.25);
}}\quad.
\ee
The SFT of the positive braid is the negative braid, as proved in Theorem 8.1 of \cite{JL} using \eqref{Equ:para isotopy}, \eqref{Equ:SF2}, \eqref{Equ:Resolution of the identity}, along with the identity
	\be
d^{-1/2}\sum_{k=0}^{d-1}q^{k\ell }\zeta^{k^{2}}=\omega\, \zeta^{-\ell^{2}}\;.
	\ee
Thus  drawing a braid at an arbitrary angle causes no confusion.

\subsubsection{The {particle-braid relation}}
\be
\raisebox{-.5cm}{
\tikz{
\fbraid{0}{0}{3\mn}{3\nn}
\node at (1.8\mn,.5\nn) {\size{$k$}};
}} =
\raisebox{-.5cm}{
\tikz{
\fbraid{0}{0}{3\mn}{3\nn}
\node at (.2\mn,2\nn) {\size{$k$}};
}}\quad.
\label{Fig:Braid-qudit}
\ee
This relation demonstrates that any charged diagram can
pass freely under (but not over) the braid. See Theorem 8.2 of \cite{JL}.

\subsubsection{Reidemeister move I}
\be\label{Positive Twist}
\raisebox{-.47cm}{
\tikz{
\draw (0,0)--(2/3,2/3);
\draw (2/3,2/3) arc (135:-135: 1.414/6);
\draw (0,1)--(1/3,2/3);
}}
=\omega^{-1/2}
\raisebox{-.47cm}{
\tikz{
\draw (0,0)--(0,1);
}}\quad.
\ee
\be\label{Negative Twist}
\raisebox{-.47cm}{
\tikz{
\draw (0,0)--(1/3,1/3);
\draw (2/3,2/3) arc (135:-135: 1.414/6);
\draw (0,1)--(2/3,1/3);
}}
=\omega^{1/2}
\raisebox{-.47cm}{
\tikz{
\draw (0,0)--(0,1);
}}\quad.
\ee

\subsubsection{Reidemeister move II}
\be
\raisebox{-.47cm}{
\tikz{
\draw (1/3,1/3)--(2/3,2/3);
\draw (2/3,2/3) arc (-45:45: 1.414/6);
\draw (2/3,1)--(1/3,1+1/3);
\draw (2/3+1/6,1/3)--(2/3,1/3+1/6);
\draw (2/3-1/6,2/3) arc (225:135: 1.414/6);
\draw (2/3+1/6,1+1/3)--(2/3,1+1/3-1/6);
}}
\ = \
\raisebox{-.47cm}{
\tikz{
\draw (0,0)--(0,1);
\draw (2/3,0)--(2/3,1);
}}\quad.
\ee

\subsubsection{Reidemeister move III}
\be
\raisebox{-.47cm}{
\tikz{
\draw (0/3,0/3)--(0/3,1/3);
\draw (0/3,2/3)--(0/3,3/3);
\draw (2/3,1/3)--(2/3,2/3);
\fbraid{0/3}{1/3}{1/3}{2/3}
\fbraid{1/3}{0/3}{2/3}{1/3}
\fbraid{1/3}{2/3}{2/3}{3/3}
}}
=
\raisebox{-.47cm}{
\tikz{
\draw (2/3,0/3)--(2/3,1/3);
\draw (2/3,2/3)--(2/3,3/3);
\draw (0/3,1/3)--(0/3,2/3);
\fbraid{1/3}{1/3}{2/3}{2/3}
\fbraid{0/3}{0/3}{1/3}{1/3}
\fbraid{0/3}{2/3}{1/3}{3/3}
}}\quad.
\ee


\subsection{Two string braids and local transformations} \label{Sect:Local Transformations}
From the 1-string braid constructed in \S \ref{Sect:Braid}, we obtain a positive and a negative 2-string braid,
\be
\raisebox{-.5cm}{\begin{tikzpicture}
\draw (1,0)--(0,1);
\draw (3/2,0)--(1/2,1);
\drawWL {}{0,0}{1,1};
\drawWL {}{1/2,0}{3/2,1};
\node at (3,1/2) {and};
\node at (6,1/2) {.};
\begin{scope}[shift={(4,1)},yscale=-1,xscale=1]
\draw (1,0)--(0,1);
\draw (3/2,0)--(1/2,1);
\drawWL {}{0,0}{1,1};
\drawWL {}{1/2,0}{3/2,1};
\end{scope}
\end{tikzpicture}}
\ee

Theoretically, there are $d$ different 2-string braids.
Their actions on 2-qudits are defined by $b_m \ket{k,l}=q^{mkl}\ket{l,k}$, $m\in \mathbb{Z}_d$.
Any neutral element can move over and under any two-string braid.
Our positive braid is $b_{-1}$ and our negative braid is $b_{1}$. Their interpolation $b_0$ is invariant under the 2-string rotation and adjoint operation, thus we represent this operator as
\be
\raisebox{-.5cm}{
\begin{tikzpicture}
\draw (1,0)--(0,1);
\draw (3/2,0)--(1/2,1);
\draw (0,0)--(1,1);
\draw (1/2,0)--(3/2,1);
\end{tikzpicture}}\;.
\ee
Since $b_0^2=I$, we call it a symmetry.
We use the symmetry $b_0$ to swap the order of the qudits.

Now we give a pictorial representation for local transformations.
Suppose Alice and Bob share a 5-qudit. The $1^{\rm st}$,  $3^{\rm rd}$, and $5^{\rm th}$ qudits belong to  Alice, and the $2^{\rm nd}$ and $4^{\rm th}$ qudits belong to Bob.
Now Bob wants to apply a transformation $T$ on his 2-qudits, namely a local transformation. Then we can represent the transformation pictorialally as follows,
\be \label{Equ:LocalTransformation}
\raisebox{-1.4cm}{\begin{tikzpicture}
\begin{scope}[yscale=1,xscale=-1]
\draw (0,0)--(0,3);
\draw (.5,0)--(.5,3);
\draw (1,0)--(1,1);
\draw (1,2)--(1,3);
\draw (1.5,0)--(1.5,1);
\draw (1.5,2)--(1.5,3);
\draw (2,0)--(3,1)--(3,2)--(2,3);
\draw (2.5,0)--(3.5,1)--(3.5,2)--(2.5,3);
\draw (3,0)--(2,1);
\draw (2,2)--(3,3);
\draw (3.5,0)--(2.5,1);
\draw (2.5,2)--(3.5,3);
\draw (4,0)--(4,3);
\draw (4.5,0)--(4.5,3);
\draw (0.8,1) rectangle (2.7,2);
\node at (1.75,1.5) {$T$};
\end{scope}
\end{tikzpicture}}\;.
\ee
One can generalize this representation for arbitrary cases, meaning any number of qudits and any partition.
Therefore we obtain a pictorial representation for local transformations using the symmetry $b_0$.
As examples of local transformations,  recall the Jordan-Wigner transformation for Pauli $X,Y,Z$ in terms of the symmetry $b_0$.

\begin{proposition}\label{Prop:JWasLocal}
When a charged string passes the symmetry $b_0$, we have the following relations: 
\be
\raisebox{-1.4cm}{
\begin{tikzpicture}
\begin{scope}[yscale=1,xscale=-1]
\begin{scope}[shift={(-2,0)}]
\draw (2,0)--(3,1)--(3,2)--(2,3);
\draw (2.5,0)--(3.5,1)--(3.5,2)--(2.5,3);
\end{scope}
\draw (2,0)--(3,1)--(3,2)--(2,3);
\draw (2.5,0)--(3.5,1)--(3.5,2)--(2.5,3);
\draw (3,0)--(0,1)--(0,2)--(3,3);
\draw (3.5,0)--(0.5,1)--(0.5,2)--(3.5,3);
\node at (.25,1.5) {$1$};
\node at (1.25,.25) {$\cdots$};
\node at (1.25,2.75) {$\cdots$};
\node at (2,1.5) {$\cdots$};
\end{scope}
\end{tikzpicture}
}
=
\raisebox{-.4cm}{
\tikz{
\draw (0,0)--(0,1);
\draw (-1\mn,0)--(-1\mn,1);
\node at (1/6,1/2) {1};
\draw (-2\mn,0)--(-2\mn,1);
\draw (-3\mn,0)--(-3\mn,1);
\draw (-6\mn,0)--(-6\mn,1);
\draw (-7\mn,0)--(-7\mn,1);
\node at (-4\mn,1/2) {$\cdots$};
\node at (-1.5\mn,1/2) {-1};
\node at (-2.5\mn,1/2) {1};
\node at (-5.5\mn,1/2) {-1};
\node at (-6.5\mn,1/2) {1};
}} \;,
\ee

\be
\raisebox{-1.4cm}{
\begin{tikzpicture}
\begin{scope}[yscale=1,xscale=-1]
\begin{scope}[shift={(-2,0)}]
\draw (2,0)--(3,1)--(3,2)--(2,3);
\draw (2.5,0)--(3.5,1)--(3.5,2)--(2.5,3);
\end{scope}
\draw (2,0)--(3,1)--(3,2)--(2,3);
\draw (2.5,0)--(3.5,1)--(3.5,2)--(2.5,3);
\draw (3,0)--(0,1)--(0,2)--(3,3);
\draw (3.5,0)--(0.5,1)--(0.5,2)--(3.5,3);
\node at (.75,1.5) {$-1$};
\node at (1.25,.25) {$\cdots$};
\node at (1.25,2.75) {$\cdots$};
\node at (2,1.5) {$\cdots$};
\end{scope}
\end{tikzpicture}
}
=
\raisebox{-.4cm}{
\tikz{
\draw (0,0)--(0,1);
\draw (-1\mn,0)--(-1\mn,1);
\node at (-1/6,1/2) {-1};
\draw (-2\mn,0)--(-2\mn,1);
\draw (-3\mn,0)--(-3\mn,1);
\draw (-6\mn,0)--(-6\mn,1);
\draw (-7\mn,0)--(-7\mn,1);
\node at (-4\mn,1/2) {$\cdots$};
\node at (-1.5\mn,1/2) {1};
\node at (-2.5\mn,1/2) {-1};
\node at (-5.5\mn,1/2) {1};
\node at (-6.5\mn,1/2) {-1};
}} \;,
\ee

\be
\raisebox{-1.4cm}{
\begin{tikzpicture}
\begin{scope}[yscale=1,xscale=-1]
\begin{scope}[shift={(-2,0)}]
\draw (2,0)--(3,1)--(3,2)--(2,3);
\draw (2.5,0)--(3.5,1)--(3.5,2)--(2.5,3);
\end{scope}
\draw (2,0)--(3,1)--(3,2)--(2,3);
\draw (2.5,0)--(3.5,1)--(3.5,2)--(2.5,3);
\draw (3,0)--(0,1)--(0,2)--(3,3);
\draw (3.5,0)--(0.5,1)--(0.5,2)--(3.5,3);
\node at (.75,1.5) {$1$};
\node at (.25,1.5) {$-1$};
\node at (1.25,.25) {$\cdots$};
\node at (1.25,2.75) {$\cdots$};
\node at (2,1.5) {$\cdots$};
\end{scope}
\end{tikzpicture}
}
=
\raisebox{-.4cm}{
\tikz{
\draw (0,0)--(0,1);
\draw (-1\mn,0)--(-1\mn,1);
\node at (1/6,1/2) {-1};
\node at (-1/6,1/2) {1};
\draw (-2\mn,0)--(-2\mn,1);
\draw (-3\mn,0)--(-3\mn,1);
\draw (-6\mn,0)--(-6\mn,1);
\draw (-7\mn,0)--(-7\mn,1);
\node at (-4\mn,1/2) {$\cdots$};
}} \;.
\ee

\end{proposition}

\begin{proof}
We use the Jordan-Wigner transformation, given in Equations \eqref{Equ:J-W x},\eqref{Equ:J-W y}, and \eqref{Equ:J-W z}.
\end{proof}

\setcounter{equation}{0}
\section{The String Fourier Transform and Entanglement\label{Sect:SFT}}
In \cite{JL} we gave a general definition of the string Fourier transform $\FS$ on planar diagrams. 
Here we analyze the special case of the SFT acting on $n$-qudits in the product basis. The SFT cyclically permutes the output strings.  The picture for  $\FS |\vec k\rangle$, where  $|\vec k\rangle $ is the decreasing basis state of \eqref{DecreasingBasis}, is given by the following (up to the normalization factor $d^{-n/4}$):
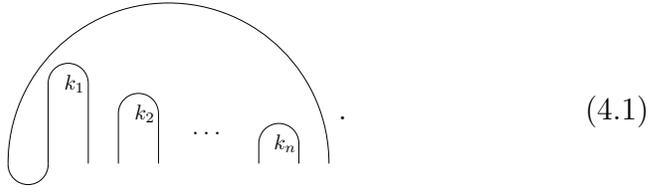
\be\label{SFT-nqudits}
\scalebox{.8}{\raisebox{-1.1cm}{
\tikz{
\fqudit{0}{0}{8\mn}{0.01\nn}{}
\fqudit{2\nn}{0}{1\mn}{4\nn}{\phantom{ll}k_{1}}
\fqudit{5.5\nn}{0}{1\mn}{2.5\nn}{\phantom{ll}k_{2}}
\fqudit{12.5\nn}{0}{1\mn}{1\nn}{\phantom{ll}k_{n}}
\fmeasure{0}{0}{-1\nn}{0.05\nn}{}
\node at (10\nn,.5) {\dots};
}}}\;.
\ee

We show in Theorem 8.2  of \cite{JL} that charges can pass freely under strings.  Therefore we can bring the string running over the top of diagram \eqref{SFT-nqudits} over the charges $k_{1}, \ldots, k_{n}$ to yield the diagram for transformation $\FS$.  The resulting transformation that acts on $|\vec k\rangle$ has 
 $2n$ input strings and $2n$ output strings, and it has charge $0$.  
It acts  on $n$-qudits as we illustrate $\FS$ in Figure~\ref{Fig:String Fourier Transform-nq}.  
\begin{figure}[h]
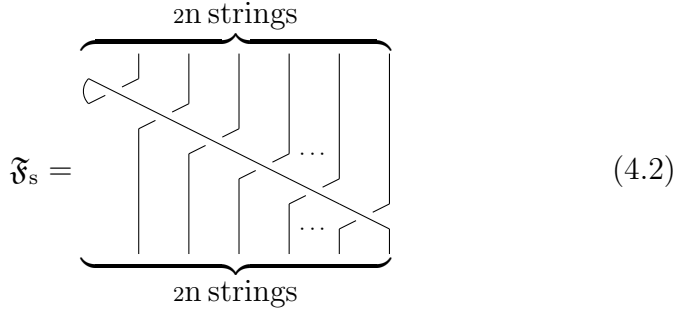

\be
\FS=\overbrace{\underbrace{\raisebox{-1cm}{\tikz{
\fbraid{4/3}{-2\nn}{2/3}{-1\nn}
\fbraid{2/3}{-1\nn}{0/3}{0\nn}
\node at (1/3,1\nn) {\tiny{$\cdots$}};
\node at (1/3,-2\nn) {\tiny{$\cdots$}};
\fbraid{-0/3}{0\nn}{-2/3}{1\nn}
\fbraid{-2/3}{1\nn}{-4/3}{2\nn}
\fbraid{-4/3}{2\nn}{-6/3}{3\nn}
\fbraid{-6/3}{3\nn}{-8/3}{4\nn}
\draw (4/3,-2\nn) -- (4/3,-3\nn);
\draw (4/3,-1\nn) -- (4/3,5\nn);
\draw (2/3,-2\nn) -- (2/3,-3\nn);
\draw (2/3,0\nn) -- (2/3,5\nn);
\draw (0/3,1\nn) -- (0/3,5\nn);
\draw (0/3,-1\nn) -- (0/3,-3\nn);
\draw (-2/3,2\nn) -- (-2/3,5\nn);
\draw (-2/3,0\nn) -- (-2/3,-3\nn);
\draw (-4/3,3\nn) -- (-4/3,5\nn);
\draw (-4/3,1\nn) -- (-4/3,-3\nn);
\draw (-6/3,4\nn) -- (-6/3,5\nn);
\draw (-6/3,2\nn) -- (-6/3,-3\nn);
\draw (-8/3,4\nn) to [bend right=45] (-8/3,3\nn);
}}}_{2\mbox{n} \; \mbox{strings}}}^{2\mbox{n} \; \mbox{strings}}
\ee
\caption{String Fourier transform on n-qudits.}\label{Fig:String Fourier Transform-nq}
\end{figure}

The $n$-qudit diagram  for $\FS$ in Figure~\ref{Fig:String Fourier Transform-nq} can also be written algebraically.  Let $b_{j,j+1,-}$ denote the negative braid in \eqref{first-neg-braid}, acting on the $j^{\rm th}$ and $(j+1)^{\rm th}$ strings. Each such transformation is local. Therefore we obtain the representation of the string Fourier transformation $\FS$ on $n$-qudits  as the product of $2n-1$ local transformations, 
\be\label{AlgebraicSFT}
\FS=\frac{1}{\sqrt{\omega}}  b_{2n-1,2n,-} \,b_{2n-2,2n-1,-}
\cdots b_{1,2,-}
= \frac{1}{\sqrt\omega}\prod_{1\leqslant j \leqslant 2n-1}^{\leftarrow} b_{j,j+1,-}\;,
\ee
with the order in the product for increasing indices from right to left. The arrow over the product sign denotes the order right to left in the product.  The remaining phase arises from  one final braid, making $2n$ braids in all composing $\FS$ on $n$-qudits. The final braid  does not appear in this explicit product, as it has been capped on the left. We apply the first Reidermeister move~\eqref{Positive Twist} to that braid, and this one obtains~$\omega^{-1/2}$.  (The picture for this move is actually a rotation of  the picture \eqref{Positive Twist} by the angle~$\pi$.)  Thus the product \eqref{AlgebraicSFT} is exactly the transformation in Figure~\ref{Fig:String Fourier Transform-nq}.

We find our topological representation of the string Fourier transform $\FS$ conceptually simpler than~\eqref{AlgebraicSFT}.  We hope that $\FS$ can also be used to advantage in quantum computation, see the general  discussion in Appendix~\ref{Sect:TQC}. 

\subsection{Matrix elements of the SFT $\FS$ on  $n$-qudits\label{Sect:SFT-Matrix}}
The SFT acts as a very interesting unitary transformation $\FS$ on the Hilbert space of $n$-qudits, that has dimension $d^{n}$.  The diagram is given above, and we can take Figure \ref{Fig:String Fourier Transform-nq} as a definition in the $2$-qudit case, with a similar picture for $n$-qudits.
We calculate the matrix elements   $ \langle\vec{\ell}\,\vert\, \FS \vert\vec{k}\rangle$ of $\FS$ in the descending qudit basis $\vert\vec{k}\rangle$ of \eqref{DecreasingBasis}, 
and the corresponding dual  basis $\langle\vec{\ell}\vert$.

\begin{theorem}\label{Thm:SFT}
The matrix elements of $\FS$ in the decreasing product basis $|\vec{k}\,\rangle$ are given algebraically by 
\be\label{SFT-Matrix-1}
\langle\vec{\ell}\,\vert\, \FS \vert\vec{k}\rangle
= d^{\frac{1-n}{2}}\,
 \omega_{\vec \ell, \vec k}\
\delta_{|\vec \ell|,|\vec k|}\;,
\quad\text{where}\quad
\omega_{\vec \ell, \vec k}
= \zeta^{|\vec{\ell}|^2} \prod_{1\leq j_1<j_2 \leq n} q^{-\ell_{j_1}k_{j_2}}\;.
\ee
Thus
	\be\label{SFT-ProductQudit}
	\FS\vert\vec{k}\rangle
	= \frac{1}{d^{\frac{n-1}{2}}} \sum_{\vec \ell:|\vec \ell |= | \vec k | }
	\omega_{\vec \ell, \vec k} \,\vert \vec{\ell} \rangle\;.
	\ee
Pictorially, the matrix elements of $d^{n/2}\FS$ are
\be\label{SFT-Matrix-0}
\raisebox{-2.3cm}{\tikz{
\fbraid{4/3}{-2\nn}{2/3}{-1\nn}
\fbraid{2/3}{-1\nn}{0/3}{0\nn}
\node at (1/3,1.5\nn) {\size{$\cdots$}};
\node at (1.1/3,-2.5\nn) {\size{$\cdots$}};
\fbraid{-0/3}{0\nn}{-2/3}{1\nn}
\fbraid{-2/3}{1\nn}{-4/3}{2\nn}
\fbraid{-4/3}{2\nn}{-6/3}{3\nn}
\fbraid{-6/3}{3\nn}{-8/3}{4\nn}
\draw (4/3,-2\nn) -- (4/3,-3\nn);
\draw (4/3,-1\nn) -- (4/3,5\nn);
\draw (2/3,-2\nn) -- (2/3,-3\nn);
\draw (2/3,0\nn) -- (2/3,5\nn);
\draw (0/3,1\nn) -- (0/3,5\nn);
\draw (0/3,-1\nn) -- (0/3,-3\nn);
\draw (-2/3,2\nn) -- (-2/3,5\nn);
\draw (-2/3,0\nn) -- (-2/3,-3\nn);
\draw (-4/3,3\nn) -- (-4/3,5\nn);
\draw (-4/3,1\nn) -- (-4/3,-3\nn);
\draw (-6/3,4\nn) -- (-6/3,5\nn);
\draw (-6/3,2\nn) -- (-6/3,-3\nn);
\draw (-8/3,4\nn) to [bend right=45] (-8/3,3\nn);
\fqudit {6\mn}{5\nn}{1\mn}{2\nn}{\phantom{ll}k_{1}}
\fqudit {2\mn}{5\nn}{1\mn}{1\nn}{\phantom{ll}k_{2}}
\fqudit {-2\mn}{5\nn}{1\mn}{0\nn}{\phantom{ll}k_{n}}
\fmeasure {6\mn}{-3\nn}{1\mn}{2\nn}{-\ell_{1}}
\fmeasure {2\mn}{-3\nn}{1\mn}{1\nn}{-\ell_{2}}
\fmeasure {-2\mn}{-3\nn}{1\mn}{0\nn}{-\ell_{n}}
}}
= \omega_{\vec \ell, \vec k}
\raisebox{-.5cm}{
\tikz{
\fqudit {0\mn}{0\nn}{1\mn}{1\nn}{\phantom{k}|\vec k|}
\fmeasure {0\mn}{0\nn}{1\mn}{1\nn}{-|\vec \ell |}
}
}\;.
\ee
\end{theorem}

\begin{proof}
We give the picture proof by para isotopy, starting from the picture for the matrix elements of Figure \ref{Fig:String Fourier Transform-nq} given by  the left side of \eqref{SFT-Matrix-0}.  Notice that there is only one continuous loop, so the diagram vanishes unless the total charge is zero.  This corresponds to the condition $|\vec k|=|\vec \ell|$, which is the same as that given by the delta function $\delta_{|\vec \ell|,|\vec k|}$.  So we need only  compute the phase of the diagram.

We  move the charges in the order $-\ell_{n}, k_{n}, -\ell_{n-1},k_{n-1}, \ldots, k_{2},-\ell_{1}$ to the left to the position of $k_{1}$.  Consider first moving  $-\ell_{n}$ to the left under the cup, producing a phase $\zeta^{\ell_{n}^{2}}$.  All further passage of $-\ell_{n}$ to  the position of $k_{1}$  by moving over caps is compensated by movement under an equal number of cups.  So that aspect of the motion produces no further phase.  Furthermore moving $-\ell_{n}$ up past each $k_{j}$ to its right produces a phase that is compensated by moving $-\ell_{n}$ down past the same $k_{j}$ still on its right. So moving $-\ell_{n}$ past the $k_{j}$'s produces no net phase.  However passing each $-\ell_{j}$ produces the net phase $q^{\ell_{n}(\sum_{1\leqslant j\leqslant n-1} \ell_{j})} $.  This completes the movement of $-\ell_{n}$.

Next move $k_{n}$ to the left to the position of  $k_{1}$.  The phases moving past caps and cups cancel. The same is true for the phases from passing other $k_{j}$'s.   We only obtain a net phase $q^{-k_{n}(\sum_{1\leqslant j \leqslant n-1}\ell_{n})}$.

Continue by moving the remaining charges to the position of $k_{1}$, in the order specified above.  We obtain the overall phase $\omega_{\vec \ell, \vec k}$. 
\end{proof}

\begin{corollary}
The matrix elements of $\FS^{-1}=\FS^{*}$ is given by
\be\label{SFT-Matrix-2}
\langle \vec \ell\, \vert \FS^{*} \vert \vec k \rangle
= d^{\frac{1-n}{2}}\,\overline{\omega}_{\vec \ell, \vec k}\
\delta_{|\vec \ell|,|\vec k|}\;.
\ee
\end{corollary}

Note that $\FS^{2n}$ is a $2\pi$ rotation, so 
\be
  \FS^{2n}\,|\vec{k}\,\rangle=q^{\abs{k}^2}\,\vert\vec{k}\rangle\;.
\ee

\subsection{The neutral $n$-qudit resource state $|{\rm Max}\rangle$}
We begin by defining an $n$-qudit resource state $\Max_{n}$.  Topologically, this state is the SFT of the product of neutral caps.  
\begin{definition}
Let  
\be
\Max_{n}=\FS | \vec{0} \rangle_{n}\;,
\quad\text{where}\quad |\vec{0}\rangle_{n}=|\underbrace{0,0,\ldots,0}_{n \text{ entries}}\rangle\;.
\ee
The pictorial representation {\rm Max}$_{n}$ of $\Max_{n}$ is  $d^{-\frac{n}{4}}$ times the diagram in Figure~\ref{Pic: Resource state}.
As earlier in this paper, we often omit the subscript $n$.
\end{definition}

\begin{figure}[h]
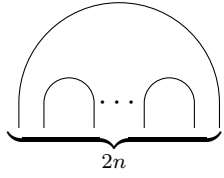

\[
\underbrace{
\raisebox{-1.5cm}{
\tikz{
\fqudit{0\mn}{0\nn}{1\mn}{1\nn}{}{}
\fqudit{-4\mn}{0\nn}{1\mn}{1\nn}{}{}
\node at (-3\mn,1\nn) {$\cdots$};
\fqudit{1\mn}{0\nn}{4\mn}{1\nn}{}{}
}}}_{2n}
\]
\caption{Pictorial representation of the multipartite entangled state $\ket{\text{Max}}$, ignoring normalization. There are $2n$ output points at the bottom.}
\label{Pic: Resource state}
\end{figure}

The standard $n$-qudit orthonormal basis $ \vert\,\vec k\,\rangle_{n}$ is characterized by a set of $n$ charges $\vec k =(k_{1},\ldots,k_{n})$, with values $k_{j}\in\Z_{d}$.  Let $\vert\vec k\vert=k_{1}+\cdots+k_{n}$ denote the total charge. 
\begin{proposition}\label{Prop:Max}
The multipartite resource state $\Max_{n}$ is  
\be \label{Maximally Entangled State}
	\Max_{n}=\frac{1}{d^{\frac{n-1}{2}}}\sum_{\vert \vec \ell\vert=0} \vert \vec \ell\,\rangle_{n}\;.
	\ee
\end{proposition}

\begin{proof}
Use  \eqref{SFT-ProductQudit}, and note that  if  $|\vec \ell |=|\vec k|=0$, then  $\omega_{\vec \ell,\vec k}=1$.
\end{proof}

\begin{remark}
The coefficients of $\vert \vec k\,\rangle_{n}$ in the sum in \eqref{Maximally Entangled State} are all positive, because
we have chosen in \eqref{DecreasingBasis} the decreasing basis $|\vec k\rangle^{\searrow}$ for our qudits. 
\end{remark}

\begin{corollary}
For $n=2$, the state $\Max_{2}$ is the generalization of the Bell state as a neutral $2$-qudit, expressed as a sum of neutral pairs,
\be\label{ResourceStateMax}
\FS\ket{0,0}=
\Max_{2}
=\frac{1}{\sqrt{d}}\sum_{\ell=0}^{d-1} \ket{\ell, -\ell}\;.
\ee
Correspondingly, one has a family of 2-qudit generalized Bell states
\be\label{OtherBellStates}
\FS \ket{k,-k} = \frac{1}{\sqrt{d}}\sum_{\ell=0}^{d-1} 
	q^{k\ell} \ket{\ell,-\ell}\;,
	\quad\text{for}\quad
	k=0, 1, \ldots, d-1\;,
\ee
arising as the SFT (and also the Fourier transform) of the neutral 2-qudit product states.
\end{corollary}

\subsection{The resource states $\Max_{n}$ and $\GHZ_{n}$\label{Sect:GHZ-Max}}
Greenberger, Horne, and Zeilinger \cite{GHZ} introduced a  multipartite resource state known as $\GHZ$ for $n$-qubit entanglement.  The corresponding $n$-qudit resource state is 
	\be\label{GHZResource}
	\GHZ_{n} =  \frac{1}{d^{\frac{1}{2}}}\sum_{k=0}^{d-1}
	\ket{k,k,\ldots,k}\;.
	\ee
Here we establish the  relation stated between $\GHZ_{n}$ and $\Max_{n}$. 

\begin{theorem}
The states $\GHZ_{n}$ and $\Max_{n}$ are Fourier duals of each other:
	\be\label{GHZ-Max}
	\GHZ_{n}
	= (F\otimes \cdots\otimes F)^{\pm1} \Max_{n}\;.
	\ee
\end{theorem}

\begin{proof}

We use the representation \eqref{Maximally Entangled State} for $\Max$. Then
\beqs
(F\otimes \cdots \otimes F) \Max_{n}
&=& \frac{1}{d^{\frac{n-1}{2} + \frac{n}{2}}}
\sum_{{ |\vec k|=0\,,\,\vec \ell}
}
q^{\vec k \cdot \vec \ell} \,\vec{\ket {\ell}}_{n}\;.
\eeqs
Insert the restriction $k_{n}=-k_{1}-\cdots -k_{n-1}$, and carry out the $(n-1)$ sums over $k_{1}, \ldots, k_{n-1}$ for fixed $\vec \ell$.  These sums vanish unless $\ell_{j}=\ell_{n}$, for each $j=1,\ldots, n-1$.  In case all the equalities hold, there are $d^{n-1}$ equal, non-zero terms.  Thus the answer is as claimed in \eqref{GHZ-Max}, namely
\be
(F\otimes \cdots \otimes F) \Max_{n}
= \frac{1}{d^{\frac{1}{2}}}
\sum_{\ell}
\ket {\ell,\ldots,\ell}_{n}=\GHZ_{n}\;.
\ee
Note that the same result arises with $F^{-1}$ in place of $F$.
\end{proof}

While the states $\Max_{n}$ and $\GHZ_{n}$ are the same state up to unitary equivalence,  there are two reasons that we prefer the state $\Max$ as the resource state for $n$-qudits:
\begin{itemize}
\item{}  The state $\Max_{n}$ has a \textit{topological interpretation}.

\item{} The state $\Max_{n}$ is \textit{charge neutral}.
\end{itemize}
\begin{remark}
Recently we have found another picture language, the three-dimensional quon language \cite{QuonLanguage}, in which $\GHZ_{n}$ and $\Max_{n}$ have both these properties.
\end{remark}

\subsection{The resource state bases $|{\rm Max}_{\vec k}\rangle$ and $|{\rm GHZ}_{\vec k} \rangle$}
We obtain a resource state basis  $|{\rm Max}_{\vec k}\rangle$ by letting the SFT act on the $n$-qudit product states $|\vec {k}\rangle$ in the basis \eqref{DecreasingBasis}.  We obtain a closed form for the resulting states using Theorem \ref{Thm:SFT}.  We obtain the following:

\begin{proposition}[\bf Entangled Resource States]
Consider the  $n$-qudit product basis elements $|\vec k\,\rangle$ defined in \eqref{DecreasingBasis}.  The SFT transforms these vectors to the generalized ${\rm Max}_{\vec k}$ basis-set,
	\be\label{SFTn-product}
		 |{\rm Max}_{\vec{k}}\rangle
		 = \FS \,| \vec k\,\rangle
		= \zeta^{-|\vec{k}|^{2}}
		 \frac{1}{d^{\frac{n-1}{2}}}
		 \sum_{\{\vec \ell: |\vec \ell |=|\vec k |\}}
		q^{k_{1}\ell_{1} + (k_{1}+k_{2})\ell_{2} +\cdots + (k_{1}+k_{2}+\cdots+k_{n})\ell_{n}}\, 
		|\vec \ell\,\rangle\;.
	\ee
Likewise the corresponding  ${\rm GHZ}_{\vec k}$ resource states are the inverse Fourier transforms of the ${\rm Max}_{\vec k}$ basis,
	\beq\label{GHZn-product}
	|{\rm GHZ}_{\vec k} \rangle
	&=& (F\otimes \cdots \otimes F)^{-1} |{\rm Max}_{\vec{k}}\rangle\\
	&=& \zeta^{-|\vec {k}|^{2}}\frac{1}{d^{\frac{1}{2}}}
	\sum_{s\in \mathbb{Z}_{d}}
	q^{-s|\vec{k}|}
	|k_{1}+s, k_{1}+k_{2}+s, \ldots, |\vec {k}|  +s   \rangle\;.\nonumber
	\eeq

\end{proposition} 
\begin{proof}
For \eqref{SFTn-product}, note that 
	\beq
	\FS |\vec k\rangle 
	&=& \sum_{\vec \ell} \langle \vec \ell|\FS |\vec k\rangle \ |\vec \ell\rangle\nonumber\\
	&=& d^{\frac{1-n}{2}}\,\sum_{\vec\ell}
\delta_{|\vec \ell|,|\vec k|}
	\zeta^{|\vec{\ell}|^2} q^{{-\ell_{1}(k_{2}+\cdots+k_{n})}-\cdots-\ell_{n-1}k_{n}}\ |\vec \ell\rangle\nonumber\\
	&=& \zeta^{|\vec{k}|^2}d^{\frac{1-n}{2}}\,\sum_{\vec\ell}
\delta_{|\vec \ell|,|\vec k|}
	 q^{\ell_{1}k_{1}-\ell_{1}|\vec k|+\cdot-\cdots+\ell_{n-1}(k_{1}+\cdots+k_{n-1})-\ell_{n-1}|\vec k|   +\ell_{n} |\vec k| - \ell_{n}|\vec k|}\ |\vec \ell\rangle\nonumber\\
	&=& \zeta^{|\vec{k}|^2}d^{\frac{1-n}{2}}\,\sum_{\vec\ell}
\delta_{|\vec \ell|,|\vec k|}
	 q^{k_{1}\ell_{1}+(k_{1}+k_{2})\ell_{2}+\cdots+ |\vec k|\ell_{n}}
	 q^{-|\vec \ell| \,|\vec k|}
	 \ |\vec \ell\rangle\nonumber\\
	&=& \zeta^{-|\vec{k}|^2}d^{\frac{1-n}{2}}\,\sum_{\vec\ell}
\delta_{|\vec \ell|,|\vec k|}
	 q^{k_{1}\ell_{1}+(k_{1}+k_{2})\ell_{2}+\cdots+ |\vec k|\ell_{n}}
	 \ |\vec \ell\rangle\nonumber\\
	 &=&
	 \zeta^{-|\vec{k}|^{2}}
		 \frac{1}{d^{\frac{n-1}{2}}}
		 \sum_{\{\vec \ell: |\vec \ell |=|\vec k |\}}
		q^{k_{1}\ell_{1} + (k_{1}+k_{2})\ell_{2} +\cdots + (k_{1}+k_{2}+\cdots+k_{n})\ell_{n}}\, 
		|\vec \ell\,\rangle\;.\nonumber
	\eeq
In order to establish \eqref{GHZn-product} use \eqref{SFTn-product} to write
\beqs
	\FS |\vec k\rangle 
	&=& \zeta^{-|\vec{k}|^2}\frac{1}{d^{\frac{n+1}{2}}}\,\sum_{\vec\ell,s}
	q^{s(|\vec \ell|-|\vec k|)}
	 q^{k_{1}\ell_{1}+(k_{1}+k_{2})\ell_{2}+\cdots+ |\vec k|\ell_{n}}
	 \ |\vec \ell\rangle\nonumber\;,
\eeqs
so 	
\beqs
	(F\otimes \cdots\otimes F)^{-1}\FS |\vec k\rangle 
	&=& \zeta^{-|\vec{k}|^2}\frac{1}{d^{\frac{2n+1}{2}}}\,\sum_{\vec w, \vec\ell,s}
	q^{-\ell_{1}w_{1}-\cdots-\ell_{n}w_{n}}
	q^{s(|\vec \ell|-|\vec k|)}
	 q^{k_{1}\ell_{1}+(k_{1}+k_{2})\ell_{2}+\cdots+ |\vec k|\ell_{n}}
	 \ |\vec w\rangle\nonumber\;.
\eeqs
Perform the $\vec \ell\in \mathbb{Z}_{d}^{n}$ sums, which give zero unless $w_{1}=k_{1}+s$, $w_{2}=k_{1}+k_{2}+s$, $\ldots,$ and $w_{n}=|\vec k|+s$. Thus
\beqs
	(F\otimes \cdots\otimes F)^{-1}\FS |\vec k\rangle 
	&=& \zeta^{-|\vec{k}|^2}\frac{1}{d^{\frac{1}{2}}}\,\sum_{s\in\mathbb{Z}_{d}}
	q^{-k|\vec k|}
	 \ | k_{1}+s, k_{1}+k_{2}+s, \ldots, |\vec k|+s\rangle\nonumber\;.
\eeqs
This is \eqref{GHZn-product} as claimed.
\end{proof}

\subsection{The red dashed line}\label{Sect:DottedLine}
We give the corresponding pictorial representation of the  $n=2$ resource
state $\Max_{2}$  in Figure~\ref{Fig:Resource state}.
\begin{center}
\begin{figure}[h]
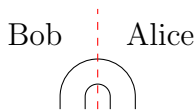

\[
\raisebox{-.5cm}{
\tikz{
\fdoublequdit{1\mn}{0}{.5\mn}{.5\nn}{}{}
\draw [red,dashed] (-.5\mn,0\nn)-- (-.5\mn,4\nn);
\node at (-3\mn,3\nn) {Alice};
\node at (2\mn,3\nn) {Bob};
}}
\]
\caption{Pictorial resource state: Only the strings in the resource state are allowed to pass the red (dashed) line between Alice and Bob. \label{Fig:Resource state}}
\end{figure}
\end{center}
In a multipartite communication protocol, 
people have different physical  locations.
We indicate these locations pictorially as the regions separated by red (dashed) lines.
Only strings of the resource state are allowed to pass these red lines in order to connect different regions.
The intersection of the strings and the red lines represents the entanglement of the resource state.  (The red dashed line is only for explanation, not a part of the protocol.)

We use $\Max$ as a resource to connect diagrams belonging to multiple persons in a quantum network. 

\begin{definition}
We say that a protocol costs one $n$-edit, when it uses one $n$-qudit $\Max$ as a resource. 
\end{definition}

\noindent{\bf Notation:}
When the multipartite entangled state $\Max$ occurs in protocols, we indicate the corresponding  $n$-qudit resource by the picture in Figure~\ref{Fig:Protocol Max}.
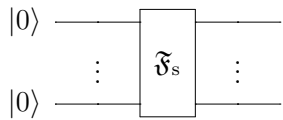
\begin{figure}[h]
\[
\scalebox{0.91}{
\Qcircuit @C=1.5em @R=2em {
\lstick{|0 \rangle} & \qw &\multigate{1}{\FS} & \qw & \qw \\
\lstick{|0 \rangle} & \ustick{\raisebox{.1cm}{\vdots}} \qw & \ghost{\FS} & \ustick{ \raisebox{.1cm}{\vdots}} \qw & \qw
}}
\]
\caption{Notation for $\ket{\text{Max}}$ in protocol, the multipartite entangled resource state. There are $n$ input and $n$ output lines.}\label{Fig:Protocol Max}
\end{figure}

\subsection{Maximal Entanglement and the SFT\label{Sect:Entropy}}
There are several possible ways to define the entanglement entropy for multi-qudits.  We give one particular definition for an $n$-qudit density matrix $\rho$, namely a positive matrix with trace one. For a vector state, denote $\rho_{\vert \vec{k}\rangle}=\vert \vec{k} \rangle\,\langle \vec {k} \vert$. Let $S$ denote an element of $\{1,2,\ldots,n\}$ and $S'$ denote its complement.  Define the entanglement entropy for the set $S$ as
	\be
	\E_{S}(\rho)
	\equiv\E(\Tr_{S'} (\rho))\;,
	\ee
where $\E(\rho)$ denotes the von Neumann entropy $-\Tr_{N^{\perp}}(\rho\ln\rho)$, where $N$ is the null space of $\rho$, and  $\Tr_{S'}$ denotes the partial trace on $S'$.  A pure state has entropy zero. Here we consider charge-zero pure product states          $|\vec{k}\rangle$ of the form \eqref{DecreasingBasis}, and the entropy of their SFTs $\FS|\vec{k}\rangle$ given in \eqref{SFTn-product}.  The following result,  independent of the choice of  $S$ and of $\vert \vec{k}\rangle$, shows that the SFT maps neutral product states to states with maximal entanglement entropy. This is the origin of the name Max. 

We consider the decreasing basis states $|\vec{k}\rangle$ defined in \eqref{DecreasingBasis}, and their SFTs $\FS |\vec k\rangle$, given pictorially up to normaliztion in \eqref{Max-kn},  and given algebraically in~\eqref{SFTn-product}.
\begin{theorem}\label{Thm:MaxEntangle}
For each single element $S$ and each basis  $n$-qudit, charge-zero  product state $\vert \vec{k} \rangle$, the SFT produces a state $\FS\vert \vec{k} \rangle$ with  entanglement entropy
	\be
		\E_{S}(\rho_{\FS\vert \vec{k}\rangle})
		=  \ln d\;.
	\ee
The constant $\ln d$ is the maximal entropy for states. 
\end{theorem}
\begin{proof}
The operator $\Tr_{S'} (\rho)$ has trace equal to one in $M_d(\mathbb{C})$, for any element $S$ of $\{1,2,\ldots,n\}$. By the concavity of the entropy, $\E(\Tr_{S'} (\rho))$ achieves its maximum $\ln d$ at $\Tr_{S'} (\rho)=\frac{1}{d}$. By Theorem \ref{Thm:SFT}, $\Tr_{S'}(\rho_{\FS\vert \vec{k}\rangle})=\frac{1}{d}$, so $\E_{S}(\rho_{\FS\vert \vec{k}\rangle})= \ln d$.
\end{proof}

\begin{corollary}[\bf Neutral Maximally Entangled Basis]
 Each vector in the neutral Max basis set~$\FS |\vec k\rangle$, including  the state $\Max=\FS | \vec{0} \rangle$, 
achieves the maximal entanglement entropy $\ln d$. 
\end{corollary}

\subsection{The Clifford group}
Here we discuss the relation between the string Fourier transform $\FS$ acting on $n$-qudits and the Clifford group, with some additional comments in Appendix~\ref{Sect:TQC}.

It is known that the $1$-qudit transformations $X,Y,Z,F,G$ with  $C_{Z}$ generate the $n$-qudit Clifford group; see Theorem 7 of \cite{Farinholt}.
When $n=1$, we infer from \eqref{Equ:SF1}, \eqref{Fig:Braid-qudit}, and \eqref{Positive Twist}, that
\be
\omega^{1/2}
\raisebox{-.3cm}{
\tikz{
\fbraid{0\mn}{0}{2\mn}{2\nn}
}}
=\FS=G.
\ee
We conclude that $\FS$ on 1-qudits is in the Clifford group. 

In the $n=2$ case,  $\FS$ is a $d^{2}\times d^{2}$ matrix. This matrix is block-diagonal, as it preserves the $d$ different $2$-qudit subspaces of fixed total charge, each of dimension $d$.

\begin{theorem}\label{Thm:FSClifford1}
For the $\FS$ on $2$-qudits, we have that
\be
C_Z=(GF^{-1}\otimes FG^{-1})\FS (1\otimes F^{-1}G^{-1}).
\ee
Thus $X,Y,Z,F,G$ and $\FS$ generate the Clifford group.
\end{theorem}

\begin{proof}
One can check the equality by computing the matrix elements on both sides. 
\end{proof}

\begin{theorem}\label{Thm:FSClifford2}
The transformation $\FS$ on $n$-qudits is in the $n$-qudit Clifford group.
\end{theorem}

\begin{proof}
We have shown that the negative braid
$$
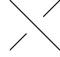
\tikz{\fbraid{-2\mn}{0}{0\mn}{2\nn}}$$
acts on a qudit basis $\ket{\ell}$ as a local transformation $\omega^{-1/2}G$.
It acts on the second and third strings of a 2-qudit as
\be\label{b23-braid}
b_{2,3,-}
= \raisebox{-.3cm}{
\tikz{
\fbraid{-2\mn}{0}{-0\mn}{2\nn}
\draw (2\mn,0)--(2\mn,2\nn);
\draw (-4\mn,0)--(-4\mn,2\nn);
}}\;.
\ee
Then
\begin{align}
b_{2,3,-}&=\omega(1 \otimes G^{-1})\FS (G^{-1}\otimes 1)\;.
\end{align}
Thus $b_{2,3,-}$ is in the Clifford group. Therefore $\FS$ on $n$-qudits is in the the Clifford group.
\end{proof}

There is another interesting formula to represent $\FS$ on 2-qudits:

\begin{proposition}\label{Prop:SFT2}
For the $\FS$ on $2$-qudits, we have that
\begin{align}
\FS&=(G^{-1} \otimes G) C_{1,X}^{-1} (F\otimes 1) C_{1,X} 
= C_{X,1}^{-1} (1\otimes F) C_{X,1} (G \otimes G^{-1})\;.
\end{align}
\end{proposition}

\begin{proof}
This follows from a direct computation.
\end{proof}

Note that $G^{-1} \otimes G$ is identity on 0-charge 2-qudits, so
\be
\FS \ket{0,0}=C_{1,X}^{-1} (F\otimes 1) \ket{0,0} \;.
\ee
The right side of this expression is the original formula for the Bell state in terms of Hadamard and CNOT.

\subsection{Measurement dictionary II}
We give a dual 2-qudit as a double-cup diagram in  \eqref{Equ:dual 2-qudit 1}, \eqref{Equ:dual 2-qudit 2}.
By Proposition \ref{Prop:SFT2}, the two corresponding protocols are given in Figure~\ref{Fig:Measure phase space 1}, \ref{Fig:Measure phase space 2} depending on the choice of the control qudit. They are equivalent to the protocol for measurement in phase space. This measurement is the most common measurement in protocols, known as the Bell state measurement for the qubit case. Thus we recover the measurement from its topological structure.  This is the reverse of the historical route to go from the algebraic measurement to its topology.
\begin{align}
\raisebox{-.65cm}{
\tikz{
\fdoublemeasure{8\mn}{4\nn}{1\mn}{1\nn}{-\ell_1}{-\ell_2}
}}
&=
\raisebox{-1.6cm}{
\tikz{
\fbraid{2\mn}{0}{0\mn}{2\nn}
\fbraid{-2\mn}{4\nn}{0\mn}{2\nn}
\draw (-2\mn,0)--(-2\mn,2\nn);
\draw (2\mn,2\nn)--(2\mn,4\nn);
\draw (4\mn,0)--(4\mn,4\nn);
\fmeasure{0\mn}{0\nn}{1\mn}{1\nn}{-\ell_1}
\fmeasure{4\mn}{0\nn}{1\mn}{3\nn}{-\ell_2}
\draw [blue,dashed] (5\mn,2\nn)-- (-3\mn,2\nn);
\draw [blue,dashed] (5\mn,0\nn)-- (-3\mn,0\nn);
}}\label{Equ:dual 2-qudit 1}\\
&=
\raisebox{-1.6cm}{
\tikz{
\fbraid{0\mn}{0}{2\mn}{2\nn}
\fbraid{2\mn}{2\nn}{4\mn}{4\nn}
\draw (4\mn,0)--(4\mn,2\nn);
\draw (0\mn,2\nn)--(0\mn,4\nn);
\draw (-2\mn,0)--(-2\mn,4\nn);
\fmeasure{0\mn}{0\nn}{1\mn}{3\nn}{-\ell_2}
\fmeasure{4\mn}{0\nn}{1\mn}{1\nn}{-\ell_1}
\draw [blue,dashed] (5\mn,2\nn)-- (-3\mn,2\nn);
\draw [blue,dashed] (5\mn,0\nn)-- (-3\mn,0\nn);
}}\label{Equ:dual 2-qudit 2}\quad.
\end{align}

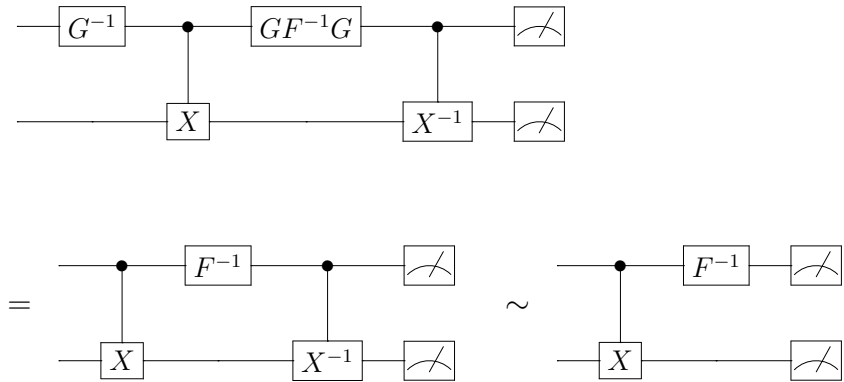
\begin{figure}[H]
\begin{align*}
&\scalebox{0.9}{
\raisebox{.7cm}{
\Qcircuit @C=1.5em @R=2em  {
& \gate{G^{-1}}                   & \control  \qw         & \gate{GF^{-1}G} &  \control \qw        & \meter &  \\
& \qw          & \gate{X} \qwx[-1] & \qw           &   \gate{X^{-1}} \qwx[-1]   & \meter  &
}}}\\
&\phantom{l} \nonumber\\
&\phantom{l} \nonumber\\
&=
\scalebox{0.9}{
\raisebox{.7cm}{
\Qcircuit @C=1.5em @R=2em  {
& \control  \qw         & \gate{F^{-1}} &  \control \qw        & \meter &  \\
& \gate{X} \qwx[-1] & \qw           &   \gate{X^{-1}} \qwx[-1]   & \meter  &
}}}
\sim 
\scalebox{0.9}{
\raisebox{.7cm}{
\Qcircuit @C=1.5em @R=2em  {
              & \control  \qw         & \gate{F^{-1}}   & \meter &  \\
        & \gate{X} \qwx[-1] & \qw             & \meter  &
}}}
\end{align*}
\caption{Measurement in the phase space: in picture language above, and translated to an algebraic protocol below.  The first algebraic protocol is the translation of the double-cup diagram on the right side of \eqref{Equ:dual 2-qudit 1}: here the measurement of the first and the second meters are $\ell_1$ and $\ell_2$ respectively. The simplification to the second algebraic protocol uses tricks in Figs.~\ref{Fig:Trick1}, \ref{Fig:Trick2}. It is equivalent to the measurement in the phase space using the trick in Figure~\ref{Fig:Trick3}.\label{Fig:Measure phase space 1}}
\end{figure}
\begin{figure}[H]
\begin{align*}
&\scalebox{0.9}{
\raisebox{.7cm}{
\Qcircuit @C=1.5em @R=2em  {
& \qw          & \gate{X} \qwx[1] & \qw           &   \gate{X^{-1}} \qwx[1]   & \meter  &\\
& \gate{G}                   & \control  \qw         & \gate{G^{-1}FG^{-1}} &  \control \qw        & \meter &
}}}\\
&\phantom{l} \nonumber\\
&=
\scalebox{0.9}{
\raisebox{.7cm}{
\Qcircuit @C=1.5em @R=2em  {
& \gate{X} \qwx[1] & \qw           &   \gate{X^{-1}} \qwx[1]   & \meter  & \\
& \control  \qw         & \gate{F} &  \control \qw        & \meter &
}}}
\sim 
\scalebox{0.9}{
\raisebox{.7cm}{
\Qcircuit @C=1.5em @R=2em  {
& \gate{X} \qwx[1] & \qw             & \meter  & \\
& \control  \qw         & \gate{F}   & \meter &
}}}
\end{align*}
\caption{Measurement in the phase space.  This protocol is a translation of \eqref{Equ:dual 2-qudit 2}.
 \label{Fig:Measure phase space 2}}
\end{figure}
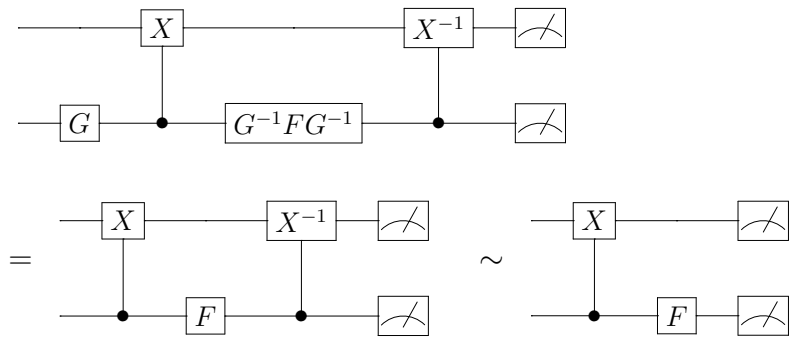

\setcounter{equation}{0}
\section{Pictorial Communication Protocols\label{Sect:Identification}}
\vbox{
Now we complete the dictionary of our holographic software.  We can use this dictionary to translate pictorial protocols to algebraic ones. When we translate between a pictorial realization of a protocol and an algebraic realization of that protocol, an overall (global) phase is irrelevant.  It does not affect a quantum-mechanical vector state, even though in this paper we often do keep track of this phase.

In this  section we  illustrate the robustness of the pictorial method, by giving examples. We identify the standard quantum teleportation protocol \cite{Bennett-etal}. As mentioned in the introduction, in a separate paper we present the new multipartite compressed teleportation (MCT) protocol \cite{ConstructiveSimulation}.}

Here we also construct an entanglement distillation protocol to produce the multipartite entangled resource state $\Max$ for $n$ persons. This protocol requires using $(n-1)$ edits, and $(n-1)$ cdits.  This cost is minimal, as is the cost in time, which is the transmission of one cdit.

\subsection{Teleportation\label{Sect:TeleportationProtocol}}
One can apply holographic software in an elementary way to ``design'' a quantum teleportation protocol.   We obtain the 1-qudit translation, that we discuss further in  \ref{sect:Topology-Algebra}.

\begin{proposition}[\bf 1-Qudit Quantum Teleportation]
 The pictorial protocol
 \begin{align}
\raisebox{-1.5cm}
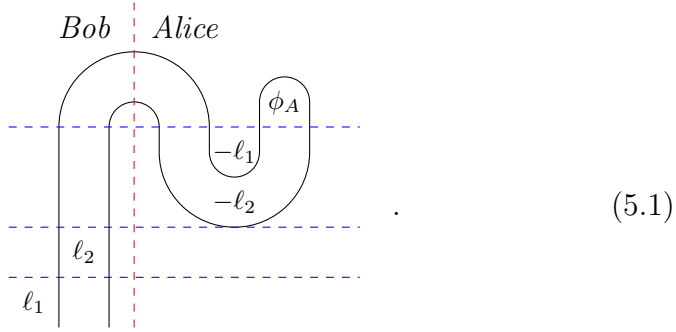
{
\tikz{
\fqudit{0}{2\nn}{1\mn}{1\nn}{\phi_A}
\fdoublequdit{8\mn}{2\nn}{1\mn}{0\nn}{}{}
\fdoublemeasure{4\mn}{2\nn}{1\mn}{1\nn}{-\ell_1}{-\ell_2}
\node at (7\mn,-3\nn) {\size{$\ell_2$}};
\node at (9\mn,-5\nn) {\size{$\ell_1$}};
\draw (8\mn,2\nn)--(8\mn,-6\nn);
\draw (6\mn,2\nn)--(6\mn,-6\nn);
\draw [blue,dashed] (10\mn,2\nn)-- (-4\mn,2\nn);
\draw [blue,dashed] (10\mn,-4\nn)-- (-4\mn,-4\nn);
\draw [blue,dashed] (10\mn,-2\nn)-- (-4\mn,-2\nn);
\draw [red,dashed] (5\mn,-6\nn)-- (5\mn,7\nn);
\node at (3\mn,6\nn) {Alice};
\node at (7\mn,6\nn) {Bob};
}}\quad.\label{Pic:Teleportation3}
\end{align}
 translates  to the algebraic protocol illustrated in Figure~\ref{Fig:Teleportation}.  For $d=2$, this reduces to the original 1-qubit teleportation protocol of Bennett et al.~\cite{Bennett-etal}.
 \end{proposition}
 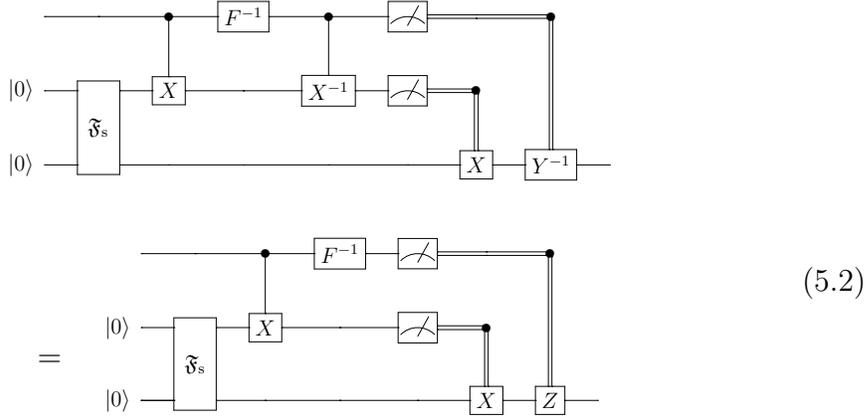
\begin{figure}[H]
 \begin{align}
&\raisebox{.6cm}{\scalebox{0.7}{
\Qcircuit @C=1.5em @R=2em {
& \qw & \control \qw & \gate{F^{-1}} & \control \qw & \meter & \cw & \control \cw \cwx[2] \\
\lstick{|0 \rangle} & \multigate{1}{\FS} & \gate{X} \qwx[-1] & \qw & \gate{X^{-1}} \qwx[-1] & \meter & \control \cw \cwx[1] \\
\lstick{|0 \rangle} & \ghost{\FS} & \qw & \qw & \qw & \qw & \gate{X} & \gate{Y^{-1}} & \qw
}}}\nonumber
\\
&&\phantom{l} \nonumber\\
\label{Fig:Teleportation-Simplified}
&\raisebox{-1cm}{=}\qquad \raisebox{.5cm}{
\scalebox{0.7}{
\Qcircuit @C=1.5em @R=2em {
& \qw & \control \qw & \gate{F^{-1}} & \meter & \cw & \control \cw \cwx[2] \\
\lstick{|0 \rangle} & \multigate{1}{\FS} & \gate{X} \qwx[-1] & \qw & \meter & \control \cw \cwx[1] \\
\lstick{|0 \rangle} & \ghost{\FS} \qw & \qw & \qw & \qw & \gate{X} & \gate{Z} & \qw
}}}
\end{align}
\caption{Teleportation protocol: 
The first protocol represents the holographic translation of the pictorial protocol  \eqref{Pic:Teleportation3}.
It simplifies to the protocol \eqref{Fig:Teleportation-Simplified} using tricks in Figs.~\ref{Fig:Trick3} and \ref{Fig:Trick4}. \label{Fig:Teleportation}
}
\end{figure}

\subsection{Multipartite resource state\label{Sect:n qudit resource state}}
We introduce the multipartite entangled resource state in \eqref{Maximally Entangled State}. We can construct this $n$-qudit resource state using $(n-1)$ of the $2$-qudit resource states.
We give the pictorial protocol in  \eqref{Pic:Max} and the algebraic protocol in Figure~\ref{Fig:Max} for the case $n=3$. One can easily generalize the protocol to the case for arbitrary~$n$.

\begin{align}
&\quad ~ \raisebox{-.8cm}{
\tikz{
\draw (8\mn,-5\nn)--(8\mn,-4\nn) [bend left=55] to (-6\mn,-4\nn)-- (-6\mn,-5\nn);
\draw (6\mn,-5\nn)--(6\mn,-4\nn) [bend left=45] to (0\mn,-4\nn) -- (0\mn,-5\nn);
\draw (-2\mn,-5\nn)--(-2\mn,-4\nn) [bend left=45] to (-4\mn,-4\nn) -- (-4\mn,-5\nn);
\draw [red,dashed] (5\mn,-5\nn)-- (5\mn,1\nn);
\draw [red,dashed] (-3\mn,-5\nn)-- (-3\mn,1\nn);
\node at (-5\mn,0\nn) {Alice};
\node at (-1\mn,0\nn) {Bob};
\node at (7\mn,0\nn) {Carol};
}} \nonumber\\
&\phantom{l} \nonumber\\
=&\raisebox{-1.5cm}{
\tikz{
\fdoublequdit{8\mn}{4\nn}{1\mn}{1\nn}{}{}
\fdoublequdit{0\mn}{4\nn}{1\mn}{1\nn}{}{}
\draw (-6\mn,4\nn)--(-6\mn,-3\nn);
\draw (-4\mn,4\nn)--(-4\mn,-3\nn);
\draw (6\mn,4\nn)--(6\mn,-3\nn);
\draw (8\mn,4\nn)--(8\mn,-3\nn);
\node at (9\mn,-2\nn) {\size{$\ell_1$}};
\fbraid{0\mn}{0}{2\mn}{2\nn}
\fbraid{2\mn}{2\nn}{4\mn}{4\nn}
\draw (4\mn,0)--(4\mn,2\nn);
\draw (0\mn,2\nn)--(0\mn,4\nn);
\draw (0\mn,-3\nn)--(0\mn,0\nn);
\draw (-2\mn,-3\nn)--(-2\mn,4\nn);
\fmeasure{4\mn}{0\nn}{1\mn}{1\nn}{-\ell_1}
\draw [red,dashed] (5\mn,-3\nn)-- (5\mn,10\nn);
\draw [red,dashed] (-3\mn,-3\nn)-- (-3\mn,10\nn);
\node at (-5\mn,9\nn) {Alice};
\node at (-1\mn,9\nn) {Bob};
\node at (7\mn,9\nn) {Carol};
\draw [blue,dashed] (9\mn,4\nn)-- (-7\mn,4\nn);
\draw [blue,dashed] (9\mn,2\nn)-- (-7\mn,2\nn);
\draw [blue,dashed] (9\mn,0\nn)-- (-7\mn,0\nn);
}}\quad.\label{Pic:Max}
\end{align}

\begin{figure}[H]
\[
\scalebox{0.9}{
\Qcircuit @C=1.5em @R=2em {
\lstick{|0 \rangle} & \multigate{1}{\FS} \qw & \qw & \qw & \qw & \qw & \qw & \qw \\
\lstick{|0 \rangle} & \ghost{\FS} \qw & \gate{X} \qwx[1]& \qw & \gate{X^{-1}} \qwx[1] &\qw &\qw & \qw \\
\lstick{|0 \rangle} & \multigate{1}{\FS} \qw & \control \qw & \gate{F} & \control \qw & \meter & \control \cw \cwx[1] & \\
\lstick{|0 \rangle} & \ghost{\FS} & \qw & \qw & \qw \qw & \qw & \gate{Y^{-1}} & \qw
}}
\]
\caption{The construction of the $n$-edit resource for $n=3$.}\label{Fig:Max}
\end{figure}

\subsection{The generalized BVK protocol\label{BVK-Protocol}}
Here we give a more general construction of $\Max$ for a multipartite network, motivated by the Bose-Vedral-Knight protocol \cite{BVK}, and the challenge of Kimble to entangle nodes across a network for a quantum internet \cite{Kimble}.  Suppose there are $n$ parties and the $j^{\rm th}$ party  has $n_{j}$ persons with a shared multipartite entangled resource state $\Max$.  In each party there is one leader who shares an extra multipartite entangled resource state $\Max$.  Then we can construct a  multipartite entangled resource state $\Max$ for all members among the $n$ parties. 

\begin{theorem}[\bf Protocol for a Multipartite Resource State]
The  multipartite, pictorial resource state for $n$ parties with $n_{j}$ persons illustrated in Equation~\eqref{BVK-picture} translates to the algebraic protocol illustrated in Equation~\eqref{BVK-Figure}. 
\end{theorem}

\begin{proof}
The algebraic protocol for the iterated construction of the multipartite entangled resource state for multipartite communication corresponding to the picture~\eqref{BVK-picture} is a straightforward application of our dictionary for holographic software to~\eqref{BVK-picture}.
\end{proof}
\beq
&&\hskip -.9cm
\raisebox{-1.3cm}{\scalebox{.8}{
\tikz{
\node at (-13\mn,8\nn) {\Large{$\cdots$}};
\draw (0\mn,-2\nn)--(0\mn,0\nn) -- (4\mn,4\nn) --(4\mn,7\nn);
\fmeasure{2\mn}{4\nn}{1\mn}{1\nn}{-\ell_n}
\node at (1\mn,-1\nn) {\size{$\ell_{1}$}};
\fqudit {-2\mn}{-2\nn}{1\mn}{6\nn}{}
\node at (-5\mn,4\nn) {\size{$\cdots$}};
\fqudit {-6\mn}{-2\nn}{1\mn}{6\nn}{}
\draw (0\mn,4\nn) [bend left=45] to (-10\mn,4\nn) -- (-10\mn,-2\nn);
\draw (2\mn,4\nn)--(2\mn,4\nn+3\nn) [bend left=45] to (-12\mn,4\nn+3\nn)--(-12\mn,4\nn);
\draw (-16\mn+0\mn,-2\nn)--(-16\mn+0\mn,1\nn) -- (-16\mn+4\mn,4\nn);
\fmeasure{-16\mn+2\mn}{4\nn+1\nn}{1\mn}{1\nn}{-\ell_2}
\node at (-15\mn,0\nn) {\size{$\ell_{n}$}};
\fqudit {-16\mn+-2\mn}{-2\nn}{1\mn}{6\nn+1\nn}{}
\node at (-16\mn+-5\mn,4\nn+1\nn) {\size{$\cdots$}};
\fqudit {-16\mn+-6\mn}{-2\nn}{1\mn}{6\nn+1\nn}{}
\draw (-16\mn+0\mn,4\nn+1\nn) [bend left=45] to (-16\mn+-10\mn,4\nn+1\nn) -- (-16\mn+-10\mn,-2\nn);
\draw (-16\mn+2\mn,4\nn+1\nn) --(-16\mn+2\mn,4\nn+1\nn+2\nn)  [bend left=45] to (-16\mn+-12\mn,4\nn+1\nn+2\nn)--(-16\mn+-12\mn,4\nn+1\nn);
\draw (-32\mn+0\mn,-2\nn) -- (-32\mn+0\mn,0\nn+2\nn) -- (-32\mn+4\mn,4\nn+1\nn);
\fmeasure{-32\mn+2\mn}{4\nn+2\nn}{1\mn}{1\nn}{-\ell_1}
\node at (-31\mn,1\nn) {\size{$\ell_{2}$}};
\fqudit {-32\mn+-2\mn}{-2\nn}{1\mn}{6\nn+2\nn}{}
\node at (-32\mn+-5\mn,4\nn+2\nn) {\size{$\cdots$}};
\fqudit {-32\mn+-6\mn}{-2\nn}{1\mn}{6\nn+2\nn}{}
\draw (-32\mn+0\mn,4\nn+2\nn) [bend left=45] to (-32\mn+-10\mn,4\nn+2\nn) -- (-32\mn+-10\mn,-2\nn);
\draw (-32\mn+2\mn,4\nn+2\nn)-- (-32\mn+2\mn,7\nn) [bend left=-45] to (4\mn,7\nn);}}} \nonumber\\
&&\phantom{x} \nonumber\\
&=&\left(\prod_{j=2}^{n}\zeta^{\ell_{j}^{2}}\right)
\ \raisebox{-.1cm}{\scalebox{.5}{
\tikz{
\draw (0\mn,0\nn) -- (0\mn,1\nn) [bend left=25] to (-42\mn,1\nn) -- (-42\mn,0\nn);
\fqudit {-2\mn}{0\nn}{1\mn}{1\nn}{}
\node at (-5\mn,1\nn) {\size{$\cdots$}};
\fqudit {-6\mn}{0\nn}{1\mn}{1\nn}{}
\draw (-10\mn,0\nn) -- (-10\mn,1\nn) [bend left=45] to (-16\mn,1\nn) -- (-16\mn,0\nn);
\fqudit {-18\mn}{0\nn}{1\mn}{1\nn}{}
\node at (-21\mn,1\nn) {\size{$\cdots$}};
\fqudit {-22\mn}{0\nn}{1\mn}{1\nn}{}
\draw (-26\mn,0\nn) -- (-26\mn,1\nn) [bend left=45] to (-32\mn,1\nn) -- (-32\mn,0\nn);
\fqudit {-34\mn}{0\nn}{1\mn}{1\nn}{}
\node at (-37\mn,1\nn) {\size{$\cdots$}};
\fqudit {-38\mn}{0\nn}{1\mn}{1\nn}{}
\node at (-13\mn,3\nn) {\Large{$\cdots$}};
}}
}\;.\label{BVK-picture}
\eeq

\beq
\label{BVK-Figure}
\raisebox{4cm}{\scalebox{0.8}{
\Qcircuit @C=1.5em @R=2em  {
& & &\lstick{\ket{0}} & \multigate{1}{\FS}  \qw & \qw & \qw & \qw  & \qw & \qw & \qw  & \qw  \\
& & \ustick{\raisebox{.2cm}{\vdots}} &\lstick{\ket{0}}  & \ghost{\FS} \qw                          & \gate{X} \qwx[1]& \qw & \gate{X^{-1}} \qwx[1] &\qw &\qw &\gate{Y^{-1}} & \qw  \\
\lstick{\ket{0}} &  \multigate{6}{\FS}  & \qw & \qw &\qw  &   \control \qw   &   \gate{F}     &  \control \qw  & \meter & \control \cw \cwx[2] & \\
&  & &\lstick{\ket{0}} & \multigate{1}{\FS}  \qw & \qw & \qw & \qw  & \qw & \qw & \qw  & \qw \\
&  &\ustick{\raisebox{.2cm}{\vdots}} &\lstick{\ket{0}}  &\ghost{\FS} \qw                         & \gate{X} \qwx[1]& \qw & \gate{X^{-1}} \qwx[1] &\qw &\gate{Y^{-1}} & \qw & \qw \\
\lstick{\ket{0}} & \ghost{\FS} & \qw & \qw &   \qw         &   \control \qw   &   \gate{F}     &  \control \qw  & \meter & \control \cw \cwx[2] & \\
\ustick{\raisebox{.1cm}{\Large{\vdots}}} &  & &\lstick{\ket{0}} & \multigate{1}{\FS}  \qw & \qw & \qw & \qw  & \ustick{\raisebox{.1cm}{\Large{\vdots}}} \qw & \qw & \qw  & \qw \\
&  &\ustick{\raisebox{.2cm}{\vdots}}  &\lstick{\ket{0}} & \ghost{\FS} \qw                         & \gate{X} \qwx[1]& \qw & \gate{X^{-1}} \qwx[1] &\qw &\gate{Y^{-1}} & \qw & \qw \\
\lstick{\ket{0}} & \ghost{\FS} & \qw & \qw & \qw      &   \control \qw   &   \gate{F}     &  \control \qw  & \meter & \cw & \control \cw \cwx[-7] &
}
}
}
\eeq

\section*{Acknowledgements}
This research was supported in part by the Templeton Religion Trust under Grants TRT0080 and TRT0159.  We are grateful for hospitality at the Research Institute for Mathematics (FIM) of the ETH-Zurich, at the Max Planck Institute for Mathematics in Bonn, at the Hausdorff Institute for Mathematics in Bonn, at the Isaac Newton Mathematical Institute in Cambridge, UK, and at the Mathematical Research Institute Oberwolfach, where we did part of this work.  We thank Erwin~Engeler, Klaus~Hepp, Daniel~Loss,  Renato~Renner, and Matthias~Troyer for helpful discussions.  We are also grateful to Bob~Coecke for sharing a copy of   \cite{Coeckebook}, before its publication.

\begin{appendix}
\renewcommand{\thesection}{A}
\renewcommand{\theequation}{A.\arabic{equation}}
\setcounter{equation}{0}
\section{}
\subsection*{Holographic software\label{Sect:Holographic}}
Let us explain in more detail what we mean by holographic software.
Quantum information protocols are expressed as products of three types of elementary operations: unitary transformations on states, measurements, and classical communication of the results of measurements.
These operations act on states, that are either inputs or given resource states.
 One can consider these operations and states as elementary network instructions.

Our pictures are composed of charged strings.   Our holographic software gives a dictionary to translate between pictures and elementary network instructions. So we refer to these pictures as elements of our software.  Our software is holographic in the sense that any algebraic protocol can be translated into pictures. This approach is helpful to simplify algebraic computations and to notice the topological features of algebraic protocols.
 We are especially interested in combinations of these instructions that we can express with elementary pictures.

Our present pictorial approach provides a way to pass in the inverse direction:  from topology to algebra. This is the major new aspect of our work.
We start from a given topological model, and use it to simulate processes in quantum information.  In doing so, we can follow topological intuition to find new concepts and new protocols.  We have resolved the technical difficulty of finding a robust and useful topological  model for quantum information in \cite{JL}.

\subsection*{Topological simulation\label{sect:Topology-Algebra}}
In another paper \cite{ConstructiveSimulation} we introduce the notion of topological simulation. Let us illustrate this concept in the simulation of  bipartite teleportation, our favorite example. We explain how we use our holographic software to recover in a natural way the resource state, measurement, and Pauli matrices---as well as the protocol of Bennett and coworkers \cite{Bennett-etal}.

One can introduce pictorial protocols and translate them into the usual algebraic protocols using our dictionary of the holographic software. 
If we simulate more complicated communication processes---by following topological intuition---then we find that our software leads to new concepts, as well as to new protocols in quantum information.

It is important to recognize what pictures have appropriate meanings in protocols.
The pictures that have virtual meanings could be used as intermediate steps in the design of protocols. 
\subsection*{Detailed Discussion of Teleportation}
Suppose that Alice wants to teleport her qudit to Bob. That means she wants to use a protocol by which the quantum information, encoded in her qudit, is faithfully delivered to Bob's location. 
We begin by drawing an elementary picture that simulates the teleportation from Alice to Bob using a noiseless channel, as  shown in \eqref{Pic:Teleportation-1}.   The red dashed line separates the communicating parties and is not a part of the picture.
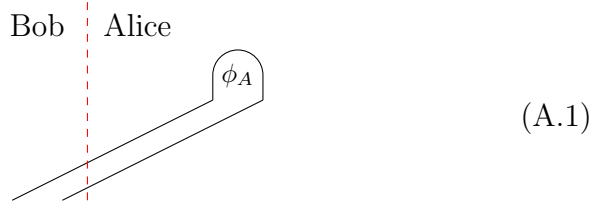
\begin{align}
&\quad \quad 
\raisebox{-1cm}{
\tikz{
\fqudit{0}{2\nn}{1\mn}{1\nn}{\phi_A}
\draw (-2\mn,2\nn)--(6\mn,-2\nn);
\draw (0\mn,2\nn)--(8\mn,-2\nn);
\draw [red,dashed] (5\mn,-2\nn)-- (5\mn,6\nn);
\node at (3\mn,5\nn) {Alice};
\node at (7\mn,5\nn) {Bob};
}}\label{Pic:Teleportation-1}\\
&\phantom{l} \nonumber\\
=&\quad \quad
\raisebox{-1cm}
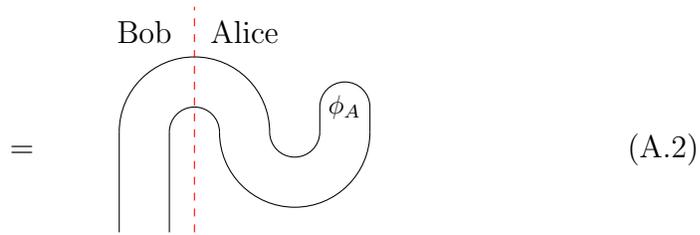
{
\tikz{
\fqudit{0}{2\nn}{1\mn}{1\nn}{\phi_A}
\fdoublequdit{8\mn}{2\nn}{1\mn}{0\nn}{}{}
\fdoublemeasure{4\mn}{2\nn}{1\mn}{0\nn}{}{}
\draw (8\mn,2\nn)--(8\mn,-2\nn);
\draw (6\mn,2\nn)--(6\mn,-2\nn);
\draw [red,dashed] (5\mn,-2\nn)-- (5\mn,7\nn);
\node at (3\mn,6\nn) {Alice};
\node at (7\mn,6\nn) {Bob};
}}\label{Pic:Teleportation-2}\\
&\phantom{l} \nonumber\\
=&
\raisebox{-1.5cm}
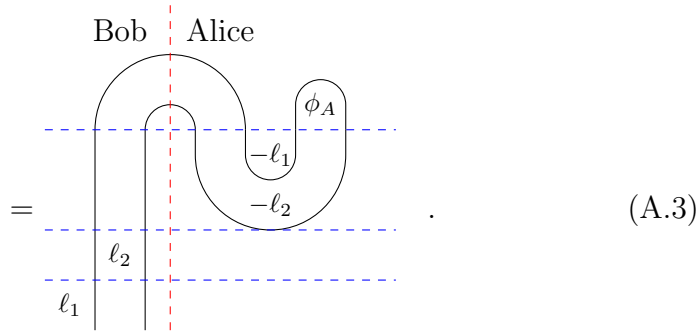
{
\tikz{
\fqudit{0}{2\nn}{1\mn}{1\nn}{\phi_A}
\fdoublequdit{8\mn}{2\nn}{1\mn}{0\nn}{}{}
\fdoublemeasure{4\mn}{2\nn}{1\mn}{1\nn}{-\ell_1}{-\ell_2}
\node at (7\mn,-3\nn) {\size{$\ell_2$}};
\node at (9\mn,-5\nn) {\size{$\ell_1$}};
\draw (8\mn,2\nn)--(8\mn,-6\nn);
\draw (6\mn,2\nn)--(6\mn,-6\nn);
\draw [blue,dashed] (10\mn,2\nn)-- (-4\mn,2\nn);
\draw [blue,dashed] (10\mn,-4\nn)-- (-4\mn,-4\nn);
\draw [blue,dashed] (10\mn,-2\nn)-- (-4\mn,-2\nn);
\draw [red,dashed] (5\mn,-6\nn)-- (5\mn,7\nn);
\node at (3\mn,6\nn) {Alice};
\node at (7\mn,6\nn) {Bob};
}}\quad.\label{Pic:Teleportation-3}
\end{align}

The disadvantage of the process \eqref{Pic:Teleportation-1} is that the noiseless channel is extremely expensive.  And the information transmitted in this way may be intercepted by an adversary.  Do we have a better simulation without the use of noiseless channels?  Yes, the solution is to use topological isotopy, which does not change the function of the protocol, but it changes the way to implement it.

The solution to this problem is to simply make a topological isotopy that deforms the picture into~\eqref{Pic:Teleportation-2}. What is the difference? Now the picture extending over the red line is on the top of the picture. The double cap that extends over the red dashed line can be implemented by an entangled state shared by Alice and Bob.

This entangled state is a resource state in our protocol, which can be realized in quantum mechanics. One obtains a first estimate of the cost of the resource states in the protocol by simply counting one-half of the number of strings over the red dashed line.

Now the caps in~\eqref{Pic:Teleportation-2} are a tensor product of two states. But what does the rest of the picture represent?  The double cup should be implemented by a measurement. In quantum information we cannot predict the outcome of the measurement. So to indicate the different possible outcomes we introduce the notion of {\it charge} on the strings.

The charges on the cups in~\eqref{Pic:Teleportation-3} indicate the result of the measurements performed by Alice. One needs to make up the opposite charge on the strings on the left, so that the function of the picture does not change. That means Alice needs to communicate  her measurement results classically to Bob, and Bob implements the corresponding recovery transformation to obtain a perfect replica at his site of the qudit that Alice teleports.

Now that  we have obtained the pictorial protocol for qudit teleportation, let us translate this with our holographic software into the  usual algebraic form. We will see how to recover from holographic software some fundamental concepts in quantum information:

Firstly, the entangled state expressed as the double cap is the standard resource state, namely the Bell state.
Secondly, the measurement expressed as the charged double cup is the measurement in phase space. In addition, measurement arising in this pictorial way maps pure states to pure states.
Thirdly, the recovery transformations that arise on the left-hand strings are Pauli $X,Y,Z$ matrices, given in~\eqref{Pauli 21}.
We end up with the original algebraic teleportation protocol that was identified by Bennett et al \cite{Bennett-etal}.
%

Therefore our topological simulation through holographic software is the reverse of the usual philosophy: algebra to topology.
Instead of presupposing the notions of quantum information, we find that they arise naturally, from following the topological intuition in our model.


\subsection*{Two philosophies}\label{Sect:TwoPhilosophies}
The mathematical history of understanding the connection  between algebra and topology goes back to the early 1900's in the development of homology theory.  A quantum version of this philosophy arose thirty years ago.  A breakthrough came  when Vaughan Jones found quantum knot invariants \cite{Jon83,Jon85,Jon87}.  In this theory, one gets pictorial representations for algebraic identities and
topological invariants. This is the direction that we call: \textit{algebra to topology}.

Jones then asked the question: can these invariants be derived in a topological manner, rather than in an
algebraic way?
In other words, can one go in  the direction \textit{topology to algebra}?
Witten answered this by giving a topological interpretation to  the Jones polynomial as an expectation of a Wilson loop in a field theory with a Chern-Simons action \cite{Wit88}.
Actually Witten's picture is more general, which led to Atiyah's notion of a 3D topological quantum field theory (TQFT) \cite{Ati88}. Reshetikhin and Turaev, constructed these examples mathematically \cite{Resh-Viro}.
The  point of this story is that 3D TQFT captures the algebraic axioms of modular tensor categories. As a concrete example, this includes representation categories of quantum groups.

In quantum information, one can imagine the same two philosophies (P1) and (P2): Our approach to quantum communication captures both philosophies.

(P1) The fundamental idea to use  pictorial notation to describe tensor manipulation originated in the work of Penrose \cite{Penrose}. 
  Later  Deutsch~\cite{D} was the first person to connect the pictorial approach to  quantum information, inspired by his work on quantum Turing machines. The pictorial networks were extended by Barenco et al.~\cite{Barenco} in the development of the modern quantum circuit language, providing elementary pictures for key quantum computational operations. Later, Lafont~\cite{Lafont} connected the ideas of reversible Boolean circuits and quantum Boolean circuits with category theory, and subsequent work by many others has developed the one-string approach to quantum information theory. This work has been extensively developed by Abramsky, Coecke, their coworkers, and many others, yielding many pictorial representations of tensors and other algebraic structures in the categorical approach to quantum protocols~\cite{AC,Coeckebook} and tensor networks~\cite{Biamonte,Christandl,Osborne}.

(P2) Many pictures have topological meaning, and its importance in quantum information was recognized in the pioneering work of Kitaev, Freedman, Larsen, Wang, Kauffman, and Lomonaco~\cite{anyon computer,topological quantum,topological quantum model,KL-1,KL-2}.  People also investigated quantum computation \cite{OgburnPreskill,BombinDelgado,Nayak-ea}.   In addition, the topological models of Kitaev, Levin, and Wen provide powerful tools in quantum information~\cite{anyon computer,LW}.

\subsection*{String Fourier transform vs.\ the braid}
Both the braid and the string Fourier transform take product states to entangled states.  We discuss these alternatives.
\subsection*{Braid}
The topological approach to quantum computation became important with Kitaev's 1997 paper proposing an anyon computer---work that only appeared some five years later in print~\cite{anyon computer}. In \S6 on the arXiv, he described the braiding and fusing of anyonic excitations in a fault-tolerant way. Freedman, Kitaev, Larsen, and Wang explored braiding further~\cite{topological quantum}, motivated by the
pioneering work of Jones, Atiyah, and Witten on knots and topological field theory~\cite{Jon85, Ati88, Wit88}.

In the case $n=2$, this braid in Equation \eqref{first-pos-braid} appears in the Jones polynomial.  
For general $n$, these braids can be ``Baxterized'' to $R_\lambda$ with parameter $\lambda$ in the sense of Jones \cite{Jon91}. That means the $R_\lambda$'s are a solution to the Yang-Baxter equation in statistical physics \cite{Yang,Baxter}; the positive braid and the negative braid are given by $\lambda=\pm i \infty$. The $R_\lambda$ have been introduced by Fateev and Zamolodchikov \cite{F-Z}.  Such kinds of braid statistics in field theory and quantum Hall systems were considered extensively by Fr\"ohlich, see \cite{Froehlich,FroehlichCargese}.
Fermionic entanglement was addressed in \cite{DoubleOccupancy,correlations-fermions}.
Kauffman and Lomonaco remarked that the braid picture describes maximal entanglement \cite{KL-2}.  

\subsection*{String Fourier transform}
Originally we had thought that the fundamental way to think about entanglement of qudits lay in the topological properties of the braid, for the braid allows consideration of isotopy in three dimensions.  

But after discovering holographic software, we have come to a different understanding.
We now believe that the \textit{string Fourier transform} (SFT) that we introduced in \cite{JL} provides a robust starting point for many aspects of quantum information, including entanglement. 

Our realization of the maximally-entangled, multipartite resource state, as well as our realization of maximal entanglement, is a consequence of the SFT. It comes from the SFT of the zero particle state.  The algebraic formulas for the SFT and for the braid can be derived from one another.  But we have learned to think about entanglement in terms of the SFT.  And this provides insight into computations, and it yields simplification for  a number of quantum information protocols; it also suggests new protocols.

Our SFT arose in the context  of planar para algebras~\cite{JL}, before we understood the depth of its significance for quantum information. Geometrically, the SFT acts on pictures and gives them a partial rotation.  These pictures might represent qudits,  transformations, or measurements.  The origin of  the SFT goes back to the work of  Ocneanu in the more general context of work on subfactor theory~\cite{Ocneanu}.

\subsection*{Quantum communication}
There are many other interesting protocols, for example~\cite{GHZ,bellexperiment,NielsenChuang,SorensenMolmer,GottesmanChuang,LoockBraunstein,Zhou-etal,Eisert-etal,Huelga-etal,Reznik-etal,ZhaoWang,MikeIke,Yu-etal,LuoWang,VanMeter,Tele2015,Hutter}, and it would be nice to analyze such protocols using holographic software.
Ideas for quantum communication build upon fundamental work  carried out  by  the groups of  Zeilinger and Pan~\cite{GHZ,AZ,Pan-ea,Zeilinger-ea,Pan-1}, and Bennett et al. \cite{Bennett-etal}.  
Recently Pan and collaborators set the record for entanglement distribution distance \cite{Yin} and for teleportation distance \cite{Ren}.  These protocols are described by the picture \eqref{Pic:Teleportation3}.

\subsection*{Quantum computation\label{Sect:TQC}}
The usual algorithm to construct the GHZ state from the fiducial, zero particle state is applying Hadamard to one qubit, and then iterating the CNOT gates one by one. Recall that Max is unitary equivalent to GHZ. However, our construction of Max is obtained from a single transformation (the SFT) on the zero particle state. Namely, Max is produced by the SFT in a single computational time step. This operation quantifies the advantage of multipartite quantum computation in a topological way. 

%
%

Any transformation on $n$-qudits can be represented as a linear sum of charged string pictures. Moreover, we can do computations on these charged string pictures. Mathematically the pictorial relations in PAPPA are very helpful to understand and simplify the computations in quantum information. 

In reality, it is easy to realize compositions of pictures, but it could be difficult to realize the linear sum. Thus it is natural to ask for a minimal set of generators which are universal for quantum computation (without using linear sums).

From this point of view, a good candidate for the generating set will include charged strings and the 1-qubit transformation $F$.  
By Theorems \ref{Thm:FSClifford1}, we already obtain the Clifford group. Theorem \ref{Thm:FSClifford1} is an indication that we do need additional generators. We show that one more is enough.

The $n$-qudit Clifford group is known to be insufficient for universal quantum computation, since the set of Clifford gates is not dense in the special unitary group $SU(d^n)$ of degree $d^n$. Moreover, the Gottesman-Knill theorem states that any quantum circuit built from stabilizer operations, which consist of Clifford gates and preparation and measurement in the standard basis, is efficiently simulable on a classical computer \cite{Gottesman-Clifford,Gottesman-fault-tolerant}.

In order to construct a universal quantum computer it is sufficient to include an arbitrary non-Clifford gate to the set of Clifford gates such that the order of the qudits is a prime number. A proof of this result can be found in Appendix D of \cite{Campbell-Magic}, whereby state injection is shown to convert distilled magic states into non-Clifford gates.

We can obtain a family of universal quantum gate sets through the inclusion of any non-Clifford gate to the generating set of the $n$-qudit Clifford group, where the order of the qudits $d$ is a prime,
\be \label{gate set 1}
\{\mbox{non-Clifford gate},X,Y,Z,F,G,\FS\}.
\ee
Some examples of non-Clifford gates include the Toffoli gate and $\pi/8$ gate \cite{MikeIke}. 

An alternative construction of a universal quantum computer is derived from the entangling property of $\FS$ on $2$-qudits. The work of Bremner et al.\ \cite{Bremner et al} and Brylinski and Brylinski \cite{Brylinski} demonstrated that the inclusion of any $2$-qudit entangling gate to the set of all single qudit gates is universal for quantum computation. Therefore we obtain the following universal gate set by including $\FS$ to the set of single qudit gates,
\be \label{gate set 2}
\{\mbox{single qudit gates},\FS\}.
\ee

The quantum gate sets \eqref{gate set 1} and \eqref{gate set 2} are new to quantum computation, and the string Fourier transform is common to both sets. Combining the universal gate sets that arise from the string Fourier transform with our pictures for qudits, transformations, and measurements constitutes a new model of universal quantum computation, that we call the PAPPA model.
\end{appendix}

\end{document}